\documentclass[aps,pra,11pt,onecolumn,nofootinbib,tightenlines,superscriptaddress]{revtex4-2}

\usepackage{mathrsfs}
\usepackage{afterpage}
\usepackage{graphicx}
\usepackage{comment}
\usepackage{dsfont}
\usepackage{amsmath}
\usepackage{amsfonts}
\usepackage{amssymb}
\usepackage{braket}
\usepackage{amsthm}
\usepackage[colorlinks=true,urlcolor=blue,citecolor=blue,linkcolor=blue]{hyperref}
\usepackage{caption}
\usepackage{lipsum}
\usepackage{subcaption}
\usepackage{mwe}
\usepackage{enumerate}
\usepackage{sidecap}
\usepackage{natbib}
\usepackage{appendix}
\usepackage{multirow}
\usepackage{array}
\usepackage{physics}
\usepackage{mathtools}
\usepackage{tikz}

\newcommand{\X}{\mathcal{X}}

\newcommand{\Z}{\mathcal{Z}}
\newcommand{\R}{\mathcal{R}}
\newcommand{\pa}{\mathrm{pa}}
\newcommand{\widebar}[1]{\overline{\mkern-1mu#1\mkern-1mu}}

\newcommand{\E}{\mathbb E}
\newcommand{\spec}{\operatorname{spec}}
\newcommand{\cR}{\mathcal R}
\newcommand{\cP}{\mathcal P}
\newcommand{\pos}[1]{\left[#1\right]_{+}}

\usepackage[paperwidth=225mm,paperheight=307mm,centering,hmargin=2.45cm,vmargin=2.5cm]{geometry}

\newtheorem{theorem}{Theorem}
\newtheorem{lemma}[theorem]{Lemma}
\newtheorem{proposition}[theorem]{Proposition}
\newtheorem{corollary}[theorem]{Corollary}
\newtheorem{definition}[theorem]{Definition}
\newtheorem{remark}[theorem]{Remark}
\let\oldremark\remark
\renewcommand{\remark}{\oldremark\upshape}

\newcommand{\mc}{\mathcal}

\renewcommand{\P}{\mathcal{P}}

\def \d {\mathrm{d}}

\DeclareMathOperator{\supp}{supp}

\makeatother

\begin{document}

%\title{Strong converse exponent of quantum privacy amplification via club-sandwiched conditional entropy} 

\title{\Large The strong converse exponent of composable randomness extraction against quantum side information}

\begin{abstract}
    We find a tight characterization of the strong converse exponent for randomness extraction against quantum side information. In contrast to previous tight bounds, we employ a composable error criterion given by the fidelity (or purified distance) to a uniform distribution in product with the marginal state. Our characterization is formulated in terms of the club-sandwiched conditional entropy, recently introduced by Rubboli, Goodarzi, and Tomamichel. It extends the corresponding result for classical side information, whose strong converse exponent was first established by Berta and Yao and later recovered by Li, Li, and Yu within this entropic framework.
 This provides the first precise operational interpretation of this family of conditional entropies in the quantum setting.
\end{abstract}

\author{Roberto Rubboli}
\email{ror@math.ku.dk}
\affiliation{Department of Mathematical Sciences, University of Copenhagen, Universitetsparken 5, 2100 Denmark}

\author{Marco Tomamichel}
\affiliation{Department of Electrical and Computer Engineering,
National University of Singapore, Singapore 117583, Singapore}
\affiliation{Centre for Quantum Technologies, National University of Singapore, Singapore 117543, Singapore}

\maketitle

\section{Introduction}

Randomness extraction is the task of extracting a uniformly distributed and secret key from a classical random variable that may be partially correlated with an adversary who possesses side information, possibly of a quantum nature. The objective is to produce a uniformly distributed random variable that is effectively uncorrelated with the eavesdropper’s side information.
This task plays a central role in quantum cryptography and quantum key distribution, where the raw key obtained from quantum measurements is generally noisy and only partially secret, due to information leakage to an eavesdropper during the protocol. Randomness extraction ensures that, after suitable post-processing, the final key is secure and can be safely used for cryptographic purposes.

Formally, let us consider a  classical-quantum state 
$\rho_{XE} =
    \sum_{x \in \X} \rho(x)\, \ketbra{x}{x}_X \otimes \rho_E(x)$.
Randomness extraction is implemented by applying a function, known as hash function,
$h:\mathcal{X}\rightarrow\mathcal{Z}$
to the register $X$. This induces the state
\begin{align}
    \mathcal{R}_h(\rho_{XE})
    =
    \sum_{z\in \Z} \ketbra{z}{z}_Z \otimes \sum_{x \in h^{-1}(z)} \rho(x)\,\rho_E(x),
\end{align}
where the register $Z$ corresponds to the hashed output and is interpreted as the final key. Here, the notation $x \in h^{-1}(z)$ indicates the set of all $x\in \X$ that are mapped to $z \in \Z$ by the hash function $h$. 
The goal of privacy amplification is to ensure that the resulting key $Z$ is close to being uniformly distributed and independent of the adversary's system $E$. Ideally, the output state should approach the product state
\begin{align}
\frac{I_{\mathcal{Z}}}{|\mathcal{Z}|} \otimes \rho_E \qquad \textrm{with} \qquad \rho_E = \sum_{x \in \X} \rho(x) \rho_E(x) \,,
\end{align}
which represents a perfectly uniform key that is completely secure against the adversary.

To study the asymptotic behaviour of this task, we are provided $n$ copies of a state $\rho_{XE}$ and perform a hash function $h_n:\mathcal{X}^{\times n}\rightarrow \mathcal{Z}_n$ on the first register of the $n$ copies. In this work, we adopt as a similarity measure the Uhlmann's fidelity.
As established in~\cite{devetak2005distillation}, there exist families of hash functions such that the fidelity with respect to the ideal state converges to one if (and only if) the size of the alphabet $\mathcal{Z}_n$ grows sufficiently slowly with the number of copies. In particular, its exponential growth rate must not exceed the von Neumann conditional entropy, i.e.,
$\limsup_{n \to \infty} \frac{1}{n}\log |\mathcal{Z}_n|
    < H(X|E)_{\rho}$.
Conversely, if
\begin{align}
    \liminf_{n \to \infty} \frac{1}{n}\log |\mathcal{Z}_n|
    > H(X|E)_{\rho} \, ,
\end{align}
privacy amplification is no longer possible (with any hash function), and the fidelity with respect to the ideal state decays exponentially to zero. This behavior follows, for example, from combining the one-shot converse bounds expressed in terms of the smooth conditional min-entropy~\cite{tomamichel2013hierarchy,shen2024optimal} with the asymptotic equipartition property~\cite{tomamichel08_aep}.

The main result of this work is to establish the exact exponential rate at which Uhlmann's fidelity to the ideal state decays to zero. This rate is referred to as the strong converse exponent with respect to the fidelity.
The strong converse exponent for randomness extraction at extraction rate $R$ under fidelity is defined as 
\begin{align}
   \label{eq: error exponent introduction}
    E_{\pa}(\rho_{XE},R) := \inf\left\{\limsup_{n\rightarrow \infty}-\frac{1}{n}\log F\left(\mathcal{R}_{h_n}(\rho_{XE}^{\otimes n}),\frac{I_{\mathcal{Z}_n}}{|\mathcal{Z}_n|}\otimes \rho_E^{\otimes n}\right) \Bigg\vert\liminf_{n\rightarrow \infty}\frac{1}{n}\log{|\Z_n|\geq R}\right\}\,,
\end{align}
where the optimization is performed over all sequences of channels $\{ \R_{h_n} \}_n$ that implement the hash functions $h_n : \X^{\times n} \to \Z_n$. Here, $F(\rho,\sigma) = (\text{Tr}[(\rho^\frac{1}{2}\sigma\rho^\frac{1}{2})^\frac{1}{2}])^2$ is the Uhlmann's fidelity.
Instead of infidelity, we can express the error in terms of purified distance, defined as $P(\rho,\sigma)\coloneqq \sqrt{1-F(\rho,\sigma)}$. It is easily verified\footnote{This can be verified using the inequality $\frac{x}{2}\le 1-\sqrt{1-x}\le x$ for $x\in[0,1]$, which implies that $\frac{F(\rho,\sigma)}{2} \le 1 - P(\rho,\sigma) \le F(\rho,\sigma)$. Formally, we thus have
\begin{align}
    E_{\pa}(\rho_{XE},R)=\inf\left\{\limsup_{n\rightarrow \infty}-\frac{1}{n}\log \left(1-P\left(\mathcal{R}_{h_n}(\rho_{XE}^{\otimes n}),\frac{I_{\mathcal{Z}_n}}{|\mathcal{Z}_n|}\otimes \rho_E^{\otimes n}\right)\right) \Bigg\vert\liminf_{n\rightarrow \infty}\frac{1}{n}\log{|\Z_n|\geq R}\right\} .
\end{align}
} 
that $F$ and $1 - P$ are equivalent at the level of exponential decay. We note that we consider the fidelity (or purified distance) to a state that is uniform on $\Z_n$ and in product with the marginal state $\rho_E^{\otimes n}$. This yields a composable error criterion~\cite{portmann22_securityreview}, in contrast to an alternative error criterion used in~\cite{li2024operational} that maximizes the fidelity over all marginal states on $E$.

Our main result is a characterization of this strong converse exponent. 
\begin{theorem}
\label{thm: main error exponent}
Let $\rho_{XE}$ be a classical-quantum state and $R \geq 0$. Then, we have
    \begin{align}
         E_{\pa}(\rho_{XE},R) = \max_{\alpha \in \left[\frac{1}{2},1\right]}\frac{1-\alpha}{\alpha} \left( R-\widetilde{H}_{\alpha}^{\lambda(\alpha)}(X|E)_\rho \right) \,,
    \end{align}
    where $\lambda(\alpha) = (1-2\alpha)/(1-\alpha)$.
\end{theorem}
Moreover, in Section~\ref{sec: universal} we show that allowing for additional public randomness, in the form of a seed that is also available to the eavesdropper, the above error exponent can be achieved universally, i.e., via a protocol that is independent of the state $\rho_{XE}$.

The result is expressed in terms of a family of club-sandwiched conditional entropies. For \(\alpha \in (0,1)\) and \(\lambda< 0\), they are defined as~\cite{rubboli2024quantum}
\begin{equation}
\widetilde{H}_\alpha^{\lambda}(A|E)_{\rho} := \min_{\sigma_E}\frac{1}{1-\alpha}\log{\Tr[\left(\rho_{AE}^\frac{1}{2}\rho_E^{\frac{1-\lambda}{2}\frac{1-\alpha}{\alpha}}\sigma_E^{\lambda\frac{1-\alpha}{\alpha}}\rho_E^{\frac{1-\lambda}{2}\frac{1-\alpha}{\alpha}}\rho_{AE}^\frac{1}{2}\right)^\alpha]}
\end{equation}
Here, the infimum is taken over all density operators $\sigma_E$ on system $E$. These entropies are part of a three-parameter family of entropies recently introduced in~\cite{rubboli2024quantum} by the present authors.
The values at $\alpha=1/2$ and $\alpha=1$ are defined via pointwise limits, and in particular, for $\lambda(\alpha) = (1-2\alpha)/(1-\alpha)$, it holds that
\begin{align}
    \lim_{\alpha\to 1^-}\widetilde{H}_{\alpha}^{\lambda(\alpha)}(X|E)_{\rho}
= H(X|E)_{\rho} \,,
\end{align}
that is, the quantity converges to the von Neumann conditional entropy (see Lemma~\ref{lem: limit of lambda conditional entropy}).

 This result provides an operational interpretation for the club-sandwiched conditional entropy for a non-trivial range of parameters. In particular, it demonstrates that this quantity has a meaningful role beyond the traditional conditional Rényi entropies that are obtained in the special cases $\lambda = 0$ and $\lambda = 1$, which correspond respectively to the arrow-down and arrow-up sandwiched conditional entropies,
\begin{align}
&\widetilde{H}^\downarrow_\alpha(A|E)_{\rho} := -\widetilde{D}_\alpha(\rho_{AE}\|I_A \otimes \rho_E)\,,\quad \textrm{and} \quad \widetilde{H}^\uparrow_\alpha(A|E)_{\rho} := \sup_{\sigma_E} -\widetilde{D}_\alpha(\rho_{AE}\|I_A \otimes \sigma_E) \,. 
\end{align}
Here, $\tilde{D}_{\alpha}(\rho \| \sigma) :=
\frac{1}{\alpha-1}\log{\Tr\big[\big(\sigma^{\frac{1-\alpha}{2\alpha}}\rho \sigma^{\frac{1-\alpha}{2\alpha}}\big)^{\alpha}\big]}$ is the sandwiched R\'enyi divergence~\cite{lennert13_renyi, Wilde3}. In addition, the club-sandwiched conditional entropy is dual to the arrow-up conditional entropy based on the maximal $\alpha$-$z$ Rényi relative entropy between $1$ and $2$ (see Lemma~\ref{lem: duality relation}). This shows that the $\alpha$-$z$ relative entropies are also relevant for privacy amplification, providing an additional application beyond the one studied in~\cite{verhagen2025conditions}, where their usefulness was previously demonstrated.

\smallskip

In Figure~\ref{fig: error exponent}, we illustrate the strong converse exponent as a function of the rate for a fixed state. 
For rates below the von Neumann conditional entropy, the strong converse exponent vanishes. 
Indeed, in this regime, the fidelity converges exponentially to one, as discussed in the introduction, 
and this behavior is faithfully captured by the expression derived in Theorem~\ref{thm: main error exponent}.\footnote{Indeed, since for all $\alpha \in [1/2,1]$, for $\lambda(\alpha) = (1-2\alpha)/(1-\alpha)$, it holds that (cf.~Lemma~\ref{lem: lower bound lambda}) 
\begin{align}
    \widetilde{H}_{\alpha}^{\lambda(\alpha)}(X|E)_\rho 
\geq H(X|E)_\rho,
\end{align}
the objective function inside the maximization is non-positive in this regime. Consequently, the maximum is attained at $\alpha = 1$, where the expression evaluates to zero.} On the other hand, we see that the strong converse exponent is linear in the rate $R$ above a certain critical value. This phenomenon is further discussed in Section~\ref{sec:critical rate}.

\begin{figure}[htbp]
\centering
\begin{tikzpicture}

    \node at (0,0) {\includegraphics[width=.5\textwidth]{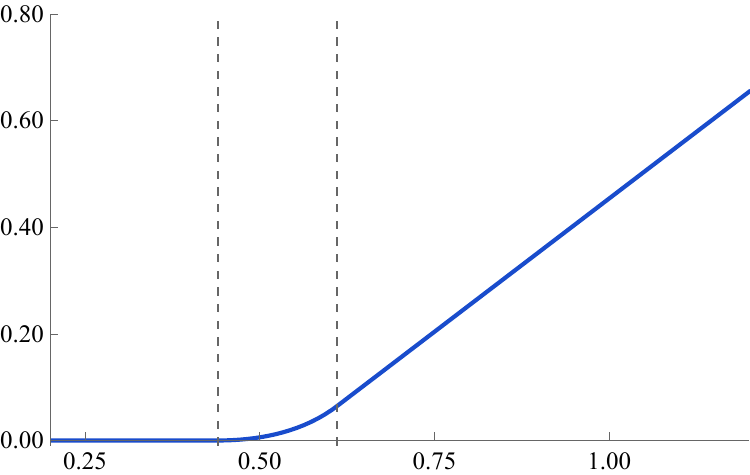}};
    
    % Axis labels
   \node at (4.5,-2.8) {\large$R$};
   \node at (-5,2.5) {\large$E_{\pa}$};
   \node at (-1.9,-2.8) {\large$H$};
    \node at (-0.4,-2.8) {\large$R_c$};

\end{tikzpicture}
\caption{The figure shows the strong converse exponent $E_{\pa}(\rho_{XE},R)$ derived in Theorem~\ref{thm: main error exponent} for the state $\rho_{XE} = 0.1\ketbra{00}{00}_{XE}+0.1\ketbra{01}{01}_{XE}+0.7\ketbra{10}{10}_{XE}+0.1\ketbra{11}{11}_{XE}$ as a function of the rate $R$. The strong converse exponent is zero when the rate is below the von Neumann conditional entropy $H(X|E)$. Indeed, in this case, the fidelity goes exponentially to one. When the rate is above the von Neumann conditional entropy $H(X|E)$, the strong converse exponent is non-zero since the fidelity decreases exponentially to zero. In addition, the strong converse exponent is linear in $R$ for the rate above a critical rate. Indeed, in this case, the optimization in Theorem~\ref{thm: main error exponent} is achieved for $\alpha=1/2$, in which case the expression becomes linear.}
    \label{fig: error exponent}
\end{figure}

\smallskip

We highlight a conceptually important feature of our characterization
of the strong converse exponent. While the randomness-extraction
criterion in~\eqref{eq: error exponent introduction} fixes the
reference marginal to $\rho_E$, the resulting expression in
Theorem~\ref{thm: main error exponent} involves conditional entropies containing both $\rho_E$ and an optimization over states $\sigma_E$.
The origin of this optimization becomes clear upon examining the
proof of Lemma~\ref{lem: one-shot converse}. There, obtaining a tight
exponent requires a careful application of H\"older's inequality,
which naturally introduces an auxiliary state $\sigma_E$ over which
the bound is optimized. This auxiliary state can be interpreted as
a ``tilting'' of the original marginal $\rho_E$, allowing us to
obtain the tightest bound. This approach has also appeared in~\cite{dupuis2023privacy,cheng2025tight}. In those cases, however, the resulting expressions did not retain both the original marginal and the
auxiliary state, giving rise to the more general club-sandwiched
conditional entropy. This analysis
highlights why this quantity arises naturally in the characterization of the
strong converse exponent. We view the identification of this
mechanism as a central technical contribution of the present work.
Beyond enabling a tight characterization of the exponent, it
provides a flexible analytic tool that may prove useful in the
study of strong converse phenomena and related exponent problems
across quantum information theory.

Finally, we note that when the rate exceeds a certain critical value, the tilting procedure is no longer needed. In this regime, the exponent can be characterized using standard methods and entropic quantities, as discussed in Section~\ref{sec:critical rate}. In our setting, the critical rate does not represent a fundamental limitation of the problem, but rather identifies the threshold beyond which our novel tilting method yields no further improvement. This interpretation becomes clear from our converse argument in Lemma~\ref{lem: one-shot converse}.
Similar critical rates have appeared in several related works~\cite{cheng2020non,li2022tight,hayashi2013tight,li2024reliability}. In the present case, it reflects the fact that there exists a threshold under which conventional techniques fail to provide tight bounds  and additional refinement is possible. Whether this phenomenon is specific to the present task or indicative of a fact that similar methods to ours could yield tight characterizations for other problems for all rates is an interesting open question.

\subsection{Prior works}
\label{sec: universal protocols}
The task of randomness extraction (or privacy amplification) has been widely studied. Several alternative definitions of the error criterion have appeared in the literature, leading to different strong converse exponents. The criterion adopted in the present work is the fidelity with the uniform state in tensor product with the marginal of the side information, i.e.,
\begin{align}
    F\left(\R_h(\rho_{XE}),\frac{I_\Z}{|I_Z|}\otimes \rho_E\right) \,.
\end{align}
This problem has been investigated in a series of papers~\cite{li2025two,berta2025strong,salzmann2022total,li2022tight}. In~\cite[Theorem~8]{berta2025strong}, the main result of the present work was first established for classical side information. We note that, although the result there is stated with an additional optimization over probability distributions, this optimization can be evaluated in closed form.\footnote{Indeed, we have that
\begin{align}
&\inf_{t}
\Bigg\{
2D(t\|p)+\frac{1-\alpha}{\alpha}\Big(r-\mathbb{E}_{y\sim t} H_\alpha(p(\cdot|y))\Big)
\Bigg\} =  \frac{1-\alpha}{\alpha} r+\inf_{t}\Big\{2D(t\|q)-2\log Z\Big\} \,,
\end{align}
where the infimum is taken over all probability distributions $t$ on the alphabet $\mathcal Y$. Here, $Z$ is a normalizing constant and $q$ is the tilted distribution defined as
\begin{equation}
Z:=\sum_y p(y)\exp\!\Big(\frac{1-\alpha}{2\alpha}H_\alpha(\cdot|y)\Big),
\qquad
q(y):=\frac{p(y)\exp\!\big(\frac{1-\alpha}{2\alpha}H_\alpha(\cdot|y)\big)}{Z}.
\end{equation}
The optimization is minimized at  $t^\star=q$. Hence, we obtain the closed form
\begin{equation}
\sup_{\alpha \in [\frac{1}{2},1]}\inf_{t}
\Bigg\{
2D(t\|p)+\frac{1-\alpha}{\alpha}\Big(r-E_{y\sim t} H_\alpha(p(\cdot|y))\Big)
\Bigg\}
=
\sup_{\alpha \in [\frac{1}{2},1]}\Bigg\{
\frac{1-\alpha}{\alpha}\,r
-2\log\sum_y p(y)\left(\sum_x p(x|y)^\alpha\right)^{\frac{1}{2\alpha}}
\Bigg\}.
\end{equation}
Hence, the result coincides with the one later derived in~\cite{li2025two}, where it is presented in a form analogous to ours, within the general entropic framework discussed in this manuscript.} 
In~\cite{li2025two}, the authors recovered the result for classical side information and reformulated it within the entropic framework introduced here. They further derive a more general statement for classical Rényi divergences as measures of distinguishability, with the fidelity arising as the special case \(\alpha = 1/2\).

In the case of quantum side information, a lower bound for the strong converse with respect to the purified distance was derived in~\cite[Theorem 3]{salzmann2022total}, given by
\begin{align}
    \sup_{\alpha \in (0,1)}(1-\alpha) \left( R-\widebar{H}_{\alpha}^\downarrow(X|E)_\rho \right) \,.
\end{align}
Our main result in Theorem~\ref{thm: main error exponent} provides a tighter lower bound together with a matching upper bound.

In~\cite{li2024operational,leditzky2016strong}, the authors consider the maximum fidelity with the ideal uniform state in tensor product with an arbitrary state on the side information. More precisely, the error criterion is given by
\begin{align}
    \max_{\sigma_E}F\left(\R_h(\rho_{XE}),\frac{I_\Z}{|Z|}\otimes \sigma_E\right) \,.
\end{align}
In particular, in~\cite{li2024operational}, the authors derive the following form for the strong converse exponent
\begin{align}
          \max_{\alpha \in [\frac{1}{2},1]}\frac{1-\alpha}{\alpha} \left( R-\widetilde{H}_{\alpha}^\uparrow(X|E)_\rho \right) \,.
    \end{align}
which includes the sandwiched conditional entropy arrow up.

Under the above criteria, the optimal deterministic function can in general depend on the underlying state \(\rho_{XE}\). To overcome this issue, privacy amplification is often implemented using a randomly chosen hash function, specified by a public seed. 
The purpose of introducing such a seed is to render the protocol universal, i.e., independent of the particular state \(\rho_{XE}\) from which secure randomness is to be extracted. 
In this case, the legitimate parties sample a seed $H$ according to a distribution $p$ on the seed space $\mathcal{H}$. The seed determines which hash function from the family $\{ h \}_{h\in \mathcal{H}}$ is applied. The parties then apply the map $\mathcal{R}_h$ corresponding to the selected hash function. The seed $H$ is revealed publicly, and
is therefore also available to the eavesdropper.

To formalize this, let $\rho_{XE}$ be the classical--quantum state
describing the raw key $X$ and the adversary's quantum side information $E$.
For each $h \in \mathcal{H}$, let $\mathcal{R}_h$ denote the
privacy-amplification map producing the key register $Z$ from $X$.
If the seed is sampled as $H \sim p$, then the joint state after
privacy amplification, including the public seed register, is
\begin{align}
\rho_{ZEH}
\;:=\;
\sum_{h\in \mathcal{H}} 
p(h)\,
\mathcal{R}_h(\rho_{XE})
\otimes 
|h\rangle\!\langle h|_{H}.
\end{align}
In the above expression, $H$ denotes a classical register containing the information about which hash function has been applied.
The corresponding ideal state is
\begin{align}
\omega_{ZEH}
\;:=\;
\frac{I_{\mathcal{Z}}}{|\mathcal{Z}|}
\otimes 
\rho_E
\otimes 
\rho_H,
\qquad
\rho_H 
:=
\sum_{h\in \mathcal{H}} 
p(h)\,
|h\rangle\!\langle h|_{H},
\end{align}
which represents a uniform key that is independent of both $E$
and the public seed $H$.
The security criterion is then naturally expressed as $F(\rho_{ZEH},\omega_{ZEH})$.
Since the seed register $H$ is classical, the square-root fidelity decomposes as an average over the seed values
\begin{align}
\label{eq: fidelity with randomness}
F(\rho_{ZEH},\omega_{ZEH})^{\frac{1}{2}}
&=
\sum_{h\in \mathcal{H}}
p(h)\,
F\!\left(
\mathcal{R}_h(\rho_{XE}),
\frac{I_{\mathcal{Z}}}{|\mathcal{Z}|}
\otimes
\rho_E
\right)^{\frac{1}{2}}=
\mathbb{E}_{\mathbb{H}\sim p}
F\!\left(
\mathcal{R}_H(\rho_{XE}),
\frac{I_{\mathcal{Z}}}{|\mathcal{Z}|}
\otimes
\rho_E
\right)^{\frac{1}{2}} .
\end{align}
Therefore, in this formulation of privacy amplification, the objective
is to maximize the expectation of the square-root fidelity.
Here, the distribution $p$ may be optimized so as to maximize the resulting fidelity. In practice, however, it is typically taken to be uniform over the seed space.
This case was studied in~\cite{hayashi2016equivocations}, where explicit expressions were derived for several error exponents, including the strong converse exponent in the classical setting for rates above the critical rate.\footnote{Note that Theorem 1 in~\cite{hayashi2016equivocations} contains a gap for rates below the critical rate, as discussed in~\cite{10403868}.} In addition, they also derived expressions for other R\'enyi relative entropies as a metric of distinguishability. In Theorem~\ref{thm: main error exponent s}, we establish the corresponding result in the fully quantum setting for all rates when fidelity is used as a metric.

In the literature, the trace distance is also frequently used as a measure of distinguishability~\cite{shen2022strong,dupuis2023privacy,renner2005universally}. The aim is to minimize the quantity
\begin{align}
\frac12\big\|\rho_{ZEH}-\omega_{ZEH}\big\|_1
= \frac12\,\mathbb{E}_{\mathbb{H}\sim p}\!\left\|
\mathcal{R}_h(\rho_{XE})-\frac{I_Z}{|Z|}\otimes \rho_E
\right\|_1 \, .
\end{align}
The trace distance is often chosen because, by the Helstrom--Holevo theorem, it admits an operational interpretation as the optimal success probability in binary state discrimination.

A lower bound for the strong converse for deterministic hash functions with respect to the trace distance was derived in~\cite[Theorem 3]{salzmann2022total}, given by
\begin{align}
    \sup_{\alpha \in [0,1]}\frac{1-\alpha}{2} \left( R-\widebar{H}_{\alpha}^\downarrow(X|E)_\rho \right) \,.
\end{align}
where $\widebar{H}_{\alpha}^\downarrow(X|E)_\rho = -\widebar{D}_\alpha(\rho_{XE}\|I_X \otimes \rho_E)$ and $\widebar{D}_\alpha(\rho\|\sigma)= \frac{1}{\alpha-1}\log{\Tr[\rho^\alpha \sigma^{1-\alpha}]}$ is the Petz relative entropy~\cite{petz86_book}.
Finally, we note that, for the specific class of strongly 
two-universal hash functions—which constitute a particular family of seeded hash functions—a lower bound on the strong converse exponent was derived in~\cite[Theorem~1]{shen2022strong}. In that setting, the bound is given by
\begin{align}
          \sup_{\alpha \in (1/2,1)}\frac{1-\alpha}{\alpha} \left( R-\widebar{H}_{2-1/\alpha}^\downarrow(X|E)_\rho \right) = \sup_{\alpha \in (0,1)}(1-\alpha) \left( R-\widebar{H}_{\alpha}^\downarrow(X|E)_\rho \right) \,.
    \end{align}

\section{Notation}
\label{sec: notation}
In this section, we introduce the notation used throughout this manuscript. We denote with $\mathcal{P}(A)$ the set of positive semidefinite operators on a Hilbert space $A$. Moreover, we denote with $\mathcal{S}(A)$ the set of quantum states, i.e., the subset of $\mathcal{P}(A)$ with unit trace.
The Schatten $p$-norm is defined as $\|A\|_p = \big(\text{Tr}[(AA^\dagger)^\frac{p}{2} ]\big)^\frac{1}{p}$ for $p \geq 1$. For any Hermitian operator $L$, we denote by $\mathcal P_L$ the pinching map associated with $L$, defined for any linear operator $K$ as
\begin{align}
    \mathcal P_L(K)
    := \sum_{\lambda \in \operatorname{spec}(L)}
    \Pi_{L,\lambda}\, K \, \Pi_{L,\lambda},
\end{align}
where $\Pi_{L,\lambda}$ denotes the orthogonal projector onto the eigenspace of $L$ corresponding to the eigenvalue $\lambda$.

The Uhlmann's fidelity between two quantum states is defined as
\begin{align}
    F(\rho,\sigma) := (\text{Tr}[(\rho^\frac{1}{2}\sigma\rho^\frac{1}{2})^\frac{1}{2}])^2\,.
\end{align}
The Umegaki relative entropy is defined as
\begin{align}
    D(\rho\|\sigma) := \Tr[\rho(\log{\rho}-\log{\sigma})] \,,
\end{align}
where one sets $D(\rho\|\sigma)=\infty$ if $\supp(\rho) \nsubseteq \supp(\sigma)$.
The sandwiched Rényi divergence~\cite{lennert13_renyi,Wilde3} and the Petz-Rényi divergence~\cite{petz86_book,tomamichel16_book}, are given by
\begin{align}
&\widetilde{D}_{\alpha}(\rho \| \sigma) :=
\frac{1}{\alpha-1}\log{\Tr\big[\big(\sigma^{\frac{1-\alpha}{2\alpha}}\rho \sigma^{\frac{1-\alpha}{2\alpha}}\big)^{\alpha}\big]} \\
&\widebar{D}_{\alpha}(\rho \| \sigma) :=
\frac{1}{\alpha-1} \log{\Tr\big[\big(\rho^\alpha \sigma^{1-\alpha}\big)\big]} \,,
\end{align}
respectively. The Log-Euclidean relative entropy is given by~\cite{hiai1993golden,mosonyi2017strong,audenaert13_alphaz}
\begin{align}
    &D_{\alpha,\infty}(\rho \| \sigma) :=
\frac{1}{\alpha-1} \log{\Tr\big[P\exp\left(P(\alpha\log \rho +(1-\alpha) \log\sigma)P\right)\big]} \,,
\end{align}
where $P$ is the projector into the intersection of the supports of $\rho$ and $\sigma$. The Log-Euclidean relative entropy admits the following useful variational characterization (see, e.g.,~\cite[Theorem~III.6]{mosonyi2017strong})
\begin{align}
\label{eq: variational form LE}
    D_{\alpha,\infty}(\rho\|\sigma) = \frac{1}{1-\alpha} \min_{\tau}\{\alpha D(\tau\|\rho)+(1-\alpha)D(\tau\|\sigma)\} \,,
\end{align}
where the minimum is taken over all states $\tau$.

\bigskip

The von Neumann conditional entropy is
\begin{align}
    H(A|E)_{\rho} := -D(\rho_{AE}\|I_A \otimes \rho_E)
\end{align}
Several variants of Rényi conditional entropy have been introduced in the literature.
Among the most important are the arrow-down and arrow-up sandwiched Rényi conditional entropies~\cite{lennert13_renyi,beigi13_sandwiched,Wilde3}, as well as the Petz Rényi conditional entropies~\cite{petz86_book,tomamichel08_aep}. These quantities are defined as
\begin{align}
\widetilde{H}^\downarrow_\alpha(A|E)_\rho
&:= -\widetilde{D}_\alpha\!\left(\rho_{AE}\,\middle\|\, I_A \otimes \rho_E\right), \qquad \widetilde{H}^\uparrow_\alpha(A|E)_\rho
:= \sup_{\sigma_E}\,
-\widetilde{D}_\alpha\!\left(\rho_{AE}\,\middle\|\, I_A \otimes \sigma_E\right), \\
\widebar{H}^\downarrow_\alpha(A|E)_\rho
&:= -\widebar{D}_\alpha\!\left(\rho_{AE}\,\middle\|\, I_A \otimes \rho_E\right), \qquad \widebar{H}^\uparrow_\alpha(A|E)_\rho
:= \sup_{\sigma_E}\,
-\widebar{D}_\alpha\!\left(\rho_{AE}\,\middle\|\, I_A \otimes \sigma_E\right) \,.
\end{align}

More recently, in~\cite{rubboli2024quantum}, the present authors introduced a three-parameter family of conditional entropies, which we refer to as $(\alpha,z,\lambda)$-conditional entropies. 
In this work, we are particularly interested in the special case $z=\alpha$, which gives rise to the
club-sandwiched conditional entropy. For \(\alpha\in(0,1)\) and \(\lambda<0\),
this quantity is defined as~\cite{rubboli2024quantum}
\begin{equation}
\widetilde{H}_\alpha^{\lambda}(A|E)_\rho:=\min_{\sigma_E}\frac{1}{1-\alpha}\log{\Tr[\left(\rho_{AE}^\frac{1}{2}I_A \otimes \Big(\rho_E^{\frac{1-\lambda}{2}\frac{1-\alpha}{\alpha}}\sigma_E^{\lambda\frac{1-\alpha}{\alpha}}\rho_E^{\frac{1-\lambda}{2}\frac{1-\alpha}{\alpha}}\Big)\rho_{AE}^\frac{1}{2}\right)^\alpha]} \,.
\end{equation} 
Here, the infimum is taken over all density operators $\sigma_E$ on system $E$.
We adopt the convention that the objective function is set to 
$\infty$ whenever $\supp(\rho_{AE})\not \subseteq \supp(I_A \otimes \sigma_E)$.
With this convention, negative powers of positive operators are understood in the sense of generalized inverses,
i.e., they are taken only on the support of the corresponding operator. This definition coincides with the limit of the objective function evaluated on full-rank states (see Appendix~\ref{app: non full-rank} for more details).  Although this regime is not the focus of our analysis, the definition can be extended to parameters $\lambda \geq 0$ and $\alpha > 1$. In this case, whether we take a maximum or a minimum over $\sigma_B$ depends on the sign of $\lambda(1-\alpha)$. If $\lambda(1-\alpha) \ge 0$, the optimization over $\sigma_B$ is taken as a maximum. If $\lambda(1-\alpha) < 0$, it is taken as a minimum (see~\cite{rubboli2024quantum}).

In the following, we adopt the standard convention of omitting the explicit identity $I_A$ whenever clear from context. 
We are particularly interested in the boundary of the data-processing region, namely in the case $\lambda(\alpha) = 1-\alpha/(1-\alpha)= (1-2\alpha)/(1-\alpha)$
\begin{equation}
\widetilde{H}_\alpha^{\lambda(\alpha)}(A|E)_\rho=\min_{\sigma_E}\frac{1}{1-\alpha}\log{\Tr[\left(\rho_{AE}^\frac{1}{2}\rho_E^{\frac{1}{2}}\sigma_E^{\frac{1-2\alpha}{\alpha}}\rho_E^{\frac{1}{2}}\rho_{AE}^\frac{1}{2}\right)^\alpha]} \,.
\end{equation}
This quantity admits the following limits (see Appendix~\ref{app: lemmas})
\begin{align}
    \lim_{\alpha \rightarrow \frac{1}{2}^+} \widetilde{H}_\alpha^{\lambda(\alpha)}(A|E)_\rho = \widetilde{H}_{\frac{1}{2}}^{\downarrow}(A|E)_\rho \,, \qquad \lim_{\alpha \rightarrow 1^-} \widetilde{H}_\alpha^{\lambda(\alpha)}(A|E)_\rho = H(A|E)_\rho \,.
\end{align}
Notably, the limit $\alpha\rightarrow 1^-$ it converges to the von Neumann conditional entropy. Below, we define the quantity $\widetilde{H}_\alpha^{\lambda(\alpha)}$ at the points $\alpha=1/2,1$ via the corresponding pointwise limits.

We will also need the $(\alpha,\lambda)$-Log-Euclidean conditional entropy (or LE conditional entropy for short), which is defined for $\alpha\in (0,1)$ and $\lambda < 0$ as
\begin{align}
\label{def: lambda LE}
&H_{\alpha,\infty}^{\lambda}(A|E)_{\rho} :=\min_{\sigma_E}\frac{1}{1-\alpha}\log \Tr [R_{AE}\, \exp \Big(\alpha\log{(\rho_{AE})}+(1-\alpha)R_{AE}\big((1-\lambda) \log{(\rho_E}) + \lambda \log{(\sigma_E})\big)R_{AE} \Big)]\,,
\end{align}
where $R_{AE}$ is the projector onto the support of $\rho_{AE}$. Here, we adopt the convention that the objective function is set to $+\infty$ whenever $\supp(\rho_{AE}) \not\subseteq \supp(I_A \otimes \sigma_E)$. With this convention, the logarithms are evaluated on the support of the states.  This definition coincides with the limit of the objective function evaluated on full-rank states (see Appendix~\ref{app: non full-rank} for more details).

\section{Critical rate}
\label{sec:critical rate}

In this section, we show that the expression for the strong converse exponent of privacy amplification given in Theorem~\ref{thm: main error exponent} simplifies when the extraction rate exceeds a critical value and becomes linear in the extraction rate $R$. Specifically, for rates above a critical rate $R_c$, the maximization over the parameter $\alpha \in [1/2,1]$ attains its maximum at the boundary point $\alpha=1/2$. In this case, the exponent reduces to
\begin{align}
    E_{\pa}(\rho_{XE},R) = R-\widetilde{H}_{\frac{1}{2}}^\downarrow(A|E)\,,
\end{align}
and thus depends only on the sandwiched conditional entropy $\widetilde{H}^\downarrow_{1/2}$, a quantity that has been extensively studied in the literature.
The above expression coincides with the result of~\cite[Theorem~1]{hayashi2016equivocations} for the strong converse exponent in the classical setting with an additional public seed, for rates above the critical rate. This is consistent with the result of Section~\ref{sec: universal}, where we show that the strong converse exponent remains unchanged in this setting.

\medskip
To show that, for rates $R$ exceeding the critical value $R_c$, the maximization over $\alpha \in [\tfrac12,1]$ in Theorem~\ref{thm: main error exponent} is achieved at $\alpha = \tfrac12$, we invoke a standard property of concave functions. Specifically, if a concave function is maximized over an interval $[a,b]$ and its left derivative at the endpoint $b$ is nonnegative, then the maximum is attained at $b$. It is useful to perform the change of variable $t=(1-\alpha)/\alpha$. For $\alpha \in [1/2,1]$, it holds that $t \in [0,1]$. Moreover, we define the function
\begin{align}
    G(t) := tR-t \widetilde{H}^{\lambda(t)}_{\alpha(t)}(X|E)_\rho \,.
\end{align}
where $\lambda(t) = (t-1)/t$ and $\alpha(t)=(1+t)^{-1}$.
The expression for the strong converse exponent in Theorem~\ref{thm: main error exponent} can be rewritten as 
\begin{align}
     E_{\pa}(\rho_{XE},R) = \max_{t \in [0,1]} G(t) \,.
\end{align}
Next, we show that $G(t)$ is a concave function. To this purpose, we need to show that $t \widetilde{H}^{\lambda(t)}_{\alpha(t)}(X|E)$ is a convex function of $t$.
\begin{lemma}
    Let $t \in [0,1]$ and $\rho_{XE}$ be a quantum state. Then, the function $t \rightarrow t \widetilde{H}^{\lambda(t)}_{\alpha(t)}(X|E)_\rho $ is convex. Here, $\lambda(t) = (t-1)/t$ and $\alpha(t)=(1+t)^{-1}$.
\end{lemma}
\begin{proof}
We use the asymptotic characterization provided in~\cite[Proposition 8.1]{rubboli2024quantum}
\begin{align}
     t \widetilde{H}^{\lambda(t)}_{\alpha(t)}(X|E)_\rho = \lim_{n\rightarrow\infty}\frac{2}{n}\log{\Big\|\rho_{A^n E^n}^\frac{1}{2}\rho_{E^n}^\frac{1}{2}\omega_{E^n}^{\frac{1}{2}(t-1)}}\Big\|_{\frac{2}{t+1}}\,.
\end{align}
Here, $\rho_{A^n E^n}=\rho_{AE}^{\otimes n}$ and $\rho_{E^n}=\rho_E^{\otimes n}$ and $\omega_{E^{n}}$ is the universal state described in Lemma~\ref{lem: universal state}.
Next, we invoke the complex interpolation result of Lemma~\ref{Hadamard}. For any $\theta \in (0,1)$, we choose the parameters
\begin{align}
    p_\theta = \frac{2}{1+(1-\theta) t_1+\theta t_2} \,, \quad p_0=\frac{2}{1+t_1} \,, \quad  p_1=\frac{2}{1+t_2}
\end{align}
and function 
\begin{align}
    F(z) = \rho_{A^nE^n}^\frac{1}{2}\rho_{E^n}^\frac{1}{2}\omega_{E^n}^{\frac{1}{2}((1-z) t_1+zt_2-1)}
\end{align}
We note that it holds the required relationship 
\begin{align}
    \frac{1}{p_\theta} = \frac{1-\theta}{p_0}+\frac{\theta}{p_1} \,.
\end{align}
In addition, since $t_1,t_2 \in [0,1]$, then $p_0,p_1 \in [1,2]$ and hence the assumptions of Lemma~\ref{Hadamard} are satisfied.
We obtain
\begin{align}
        \Big\|\rho_{A^nE^n}^\frac{1}{2}\rho_{E^n}^\frac{1}{2}\omega_{E^n}^{\frac{1}{2}((1-\theta) t_1+\theta t_2-1)}\Big\|_{\frac{2}{1+(1-\theta) t_1+\theta t_2}} \leq \Big\|\rho_{A^nE^n}^\frac{1}{2}\rho_{E^n}^\frac{1}{2}\omega_{E^n}^{\frac{1}{2}( t_1-1)}\Big\|_{\frac{2}{1+t_1}}^{1-\theta} \Big\|\rho_{A^nE^n}^\frac{1}{2}\rho_{E^n}^\frac{1}{2}\omega_{E^n}^{\frac{1}{2} (t_2-1)}\Big\|_{\frac{2}{1+t_2}}^{\theta}
    \end{align}
    Here, we used that we could omit the additional phases that involve powers of the form $\omega_{E^n}^{it}$ since the norm is unitarily invariant.
Applying the logarithm, dividing by $2/n$, and taking the limit $n \to \infty$ gives the desired inequality.
\end{proof}

The above result implies the function $G(t)$ is concave in the interval $[0,1]$. Hence, the maximum is attained at $t=1$ (or equivalently $\alpha=1/2$) if the left derivative at $1$ of $G(t)$ is positive. This condition defines the critical rate
\begin{align}
    \frac{\d}{\d t} G(t)\Big\vert_{t=1^-} \geq 0 \quad \implies \quad R \geq R_c = \frac{\d}{\d t} \left( t \widetilde{H}^{\lambda(t)}_{\alpha(t)}(X|E)_\rho \right) \Bigg\vert_{t=1^-}
\end{align}
As a corollary, we obtain that
\begin{corollary}
\label{cor: critical rate}
Let $\rho_{XE}$ be a classical-quantum state and 
\begin{align}
    R \geq R_c = \frac{\d}{\d t} \left( t \widetilde{H}^{\lambda(t)}_{\alpha(t)}(X|E)_\rho \right) \Bigg\vert_{t=1^-} \,,
\end{align}
where $\lambda(t) = (t-1)/t$ and $\alpha(t)=(1+t)^{-1}$.
Then, we have
    \begin{align}
         E_{\pa}(\rho_{XE},R)  = R-\widetilde{H}_{\frac{1}{2}}^\downarrow(X|E) \,.
    \end{align}
\end{corollary}

\section{Converse bound}

The proof of the main result is divided into two parts. 
First, we establish a lower bound on the strong converse exponent, commonly referred to as the converse bound. 
In the subsequent section, we derive a matching upper bound, corresponding to the achievability part.

\subsection{Coarse-graining inequality}

We begin by showing that applying a deterministic function to the classical register of a classical-quantum state cannot increase the club-sandwiched conditional entropy. Such a transformation may be viewed as a coarse-graining, since it can map distinct input values to the same output. This is consistent with the intuition that the uncertainty about the output of a function should be no greater than the uncertainty about its input: given an estimate of the input, an observer can always obtain a corresponding estimate of the output simply by applying the function.

This property also has an operational interpretation. Indeed, consider
the coarse-grained state $\sigma_{YE}=\mathcal{R}_f(\rho_{XE})$.
Since the definition of the strong converse exponent involves an
infimum over all hashing protocols, the exponent of $\sigma_{YE}$
is at least as large as that of $\rho_{XE}$. Indeed, for any hash function
$g_n:\mathcal{Y}^{\times n}\to\mathcal{Z}_n$ extracting randomness
from $\sigma_{YE}^{\otimes n}$, one can construct a hash function
$h_n=g_n\circ f^{\times n}$ acting on $\rho_{XE}^{\otimes n}$.
This construction preserves the output alphabet and produces exactly
the same output state and normalized ideal state, hence the same
fidelity. Consequently, coarse-graining cannot decrease the strong
converse exponent, that is,
$E_{\mathrm{pa}}(\rho_{XE},R)\leq E_{\mathrm{pa}}(\sigma_{YE},R)$.
Any closed-form characterization of this exponent must therefore
respect this monotonicity, providing a consistency requirement for
candidate entropic representations and motivating the study of the
corresponding inequality at the entropic level.

We emphasize that this entropic inequality is nontrivial. Indeed, as discussed below, numerical evidence shows that the Petz conditional entropies do not satisfy it. The following lemma will play a key role in establishing the converse bound for the exponent in Proposition~\ref{prop: converse}. It generalizes~\cite[Lemma 5.14]{tomamichel16_book}, where the analogous inequality was proved for the up-arrow and down-arrow sandwiched R\'enyi conditional entropies.
\begin{lemma}[Coarse-graining reduces the club-sandwiched conditional entropy]
\label{lem: Hashing}
For any $\alpha\geq 0$, $\lambda \in \mathbb{R}$ and any function $h:\mathcal{X}\to\mathcal{Z}$ and any classical-quantum state $\rho_{XBE}$ it holds that
    \begin{align}
\widetilde{H}_{\alpha}^\lambda(ZB|E)_{\mathcal{R}_h(\rho)} \leq \widetilde{H}_{\alpha}^\lambda(XB|E)_{\rho}
    \end{align}
\end{lemma}
\begin{proof}
Consider a state of the form $\rho_{XBE} = \sum_x \rho(x) \ketbra{x}{x}_X \otimes \rho_{BE}(x)$ where $\rho(x)$ are probabilities and $\rho_{BE}(x)$ states.
Let us denote the preimage of $z$ under $h$ as 
\begin{align}
    h^{-1}(z) = \{\, x \in \mathcal{X} : h(x) = z \,\} \,.
\end{align}
The output of the hash function is 
    \begin{align}
    \mathcal{R}_h(\rho_{XBE})=\sum_z \ketbra{z}{z}_Z \otimes \sum_{x \in h^{-1}(z)}\rho(x) \rho_{BE}(x)
\end{align}
Note that since the Hashing acts on the first register only
\begin{align}
    \Tr_Z(\mathcal{R}_h(\rho_{XBE})) = \rho_{BE}
\end{align}
Let us consider the case $\lambda(1-\alpha) <0$ as the other case is similar. We have
\begin{equation}
\widetilde{H}_{\alpha}^\lambda(ZB|E)_{\mathcal{R}_h(\rho)}=\min_{\sigma_E}\frac{1}{1-\alpha}\log{\sum_z\Tr[\left(\sigma_E^{\frac{\lambda}{2}\frac{1-\alpha}{\alpha}}\rho_{E}^{\frac{1-\lambda}{2}\frac{1-\alpha}{\alpha}}\sum_{x \in h^{-1}(z)}\rho(x)\rho_{BE}(x)\rho_E^{\frac{1-\lambda}{2}\frac{1-\alpha}{\alpha}}\sigma_E^{\frac{\lambda}{2}\frac{1-\alpha}{\alpha}}\right)^\alpha]} \,.
\end{equation}
McCarthy’s inequality~\cite{mccarthy1967cp} implies that for the power function with parameter $\alpha \in (0,1)$, one has $\Tr[(\sum_i A_i)^\alpha] \leq \sum_i\Tr(A_i^\alpha)$ for positive semidefinite matrices $A_i$ (see also~\cite[Lemma 2.1]{marwah2022uniform} for a proof). Hence, for $\alpha \in (0,1)$, it holds that
\begin{align}
    &\sum_z\Tr[\left(\sigma_E^{\frac{\lambda}{2}\frac{1-\alpha}{\alpha}}\rho_E^{\frac{1-\lambda}{2}\frac{1-\alpha}{\alpha}}\sum_{x \in h^{-1}(z)}\rho(x)\rho_{BE}(x)\rho_E^{\frac{1-\lambda}{2}\frac{1-\alpha}{\alpha}}\sigma_E^{\frac{\lambda}{2}\frac{1-\alpha}{\alpha}}\right)^\alpha]\\
    &\qquad \qquad \qquad \qquad \leq \sum_x\Tr[\left(\sigma_E^{\frac{\lambda}{2}\frac{1-\alpha}{\alpha}}\rho_E^{\frac{1-\lambda}{2}\frac{1-\alpha}{\alpha}}\rho(x)\rho_{BE}(x)\rho_E^{\frac{1-\lambda}{2}\frac{1-\alpha}{\alpha}}\sigma_E^{\frac{\lambda}{2}\frac{1-\alpha}{\alpha}}\right)^\alpha], 
\end{align}
where the latter is the term inside the logarithm of $\widetilde{H}_{\alpha}^\lambda(XB|E)_{\rho}$. Taking the logarithm of both sides, multiplying by $1/(1-\alpha)$, and then minimizing over $\sigma_E$ on both sides yields the result. If $\alpha > 1$ instead, we can use that the inequality holds in the opposite direction (see e.g.,~\cite[Theorem 1.2]{bourin2007matrix}). For $\alpha \in \{0, 1, \infty\}$ we can use a limiting argument.
\end{proof}

By specializing to the case where the function is constant, a direct application of the coarse-graining property to states with two classical registers yields
\begin{align} 
\widetilde{H}_{\alpha}^\lambda(XZ|E)_{\rho} \geq \widetilde{H}_{\alpha}^\lambda(Z|E)_{\rho} 
\label{eq:coarse-graining simplified}\,,
\end{align}
i.e., that discarding classical information cannot increase conditional entropy. In fact, this property is equivalent to coarse-graining. Indeed, one can consider an invertible (isometric) map that computes the output while retaining the input, and then use the invariance of Rényi entropies under local isometries to obtain the desired equivalence.

\begin{remark}
    Two commonly used definitions of conditional entropy are based on Petz-R\'enyi relative entropy, namely $\widebar{H}^\downarrow_\alpha(X|E)_\rho$ and $\widebar{H}^\uparrow_\alpha(X|E)_\rho$. Numerical evidence suggests that they do not always satisfy the coarse-graining inequality, not even in the special case of discarding classical information as in~\eqref{eq:coarse-graining simplified}.
\end{remark}

\subsection{A one-shot converse bound}

We first derive a one-shot bound as an application of H\"older’s inequality. The use of H\"older inequalities in converse theorems can be traced back to~\cite{leditzky2016strong}. Compared to the converse bound in~\cite{li2024operational}, we note that in our error criterion, we do not need to optimize over the marginal on $E$, which yields a converse bound that requires a more refined application of the H\"older inequality.

\begin{lemma}
\label{lem: one-shot converse}
Let $\rho_{XE}$ be a quantum state. For any $\alpha \in \left[\frac{1}{2},1\right)$, we have
    \begin{align}
        -\frac{\alpha}{1-\alpha}\log\left(F\left(\rho_{XE},\frac{I_{\mathcal{X}}}{|\mathcal{X}|}\otimes \rho_E\right)\right) \geq \log{|\mathcal{X}|}-\widetilde{H}_{\alpha}^{\lambda(\alpha)}(X|E)_\rho \,,
    \end{align}
     where $\lambda(\alpha) = (1-2\alpha)/(1-\alpha)$.
\end{lemma}
\begin{proof}
Let $\alpha \in [1/2,1)$. We define $\beta$ to be the H\"older conjugate of $\alpha$, i.e., the unique parameter satisfying
\begin{align}
    \frac{1}{2\alpha} + \frac{1}{2\beta} = 1 .
\end{align}
Equivalently, $\beta = \frac{\alpha}{2\alpha - 1}$ and hence $\beta \in (1, \infty]$. H\"older's inequality implies that for linear operators $X,Y$, it holds that $\|XY\|_1\leq \|X\|_{2\alpha}\|Y\|_{2\beta}$ (see e.g.,~\cite[Exercise IV.2.7]{bhatia1997matrix}).  
 Let $\sigma_E$ be a state such that $\supp(\rho_{AE}) \subseteq\supp(I_A \otimes \sigma_E)$. We have
    \begin{align}
        &\frac{\alpha}{1-\alpha} \log F\left(\rho_{XE},\pi_X\otimes \rho_E\right) \\
        &\qquad = \frac{2\alpha}{1-\alpha}\log\|\rho_{XE}^\frac{1}{2}(\pi_X\otimes \rho_E)^\frac{1}{2}\|_1 \\
        &\qquad =\frac{2\alpha}{1-\alpha}\log\|\rho_{XE}^\frac{1}{2} (\pi_X \otimes \rho_E)^\frac{1}{2}(\pi_X \otimes \sigma_E)^\frac{1-2\alpha}{2\alpha}(\pi_X \otimes \sigma_E)^{-\frac{1-2\alpha}{2\alpha}}\|_1 \\
        &\qquad \leq \frac{2\alpha}{1-\alpha}\log\|\rho_{XE}^\frac{1}{2} (\pi_X \otimes \rho_E)^\frac{1}{2}(\pi_X \otimes \sigma_E)^\frac{1-2\alpha}{2\alpha}\|_{2\alpha} \|(\pi_X \otimes \sigma_E)^{-\frac{1-2\alpha}{2\alpha}}\|_{2\beta} \\
        & \qquad   = \frac{2\alpha}{1-\alpha}\log\|\rho_{XE}^\frac{1}{2} (\pi_X \otimes \rho_E)^\frac{1}{2}(\pi_X \otimes \sigma_E)^\frac{1-2\alpha}{2\alpha}\|_{2\alpha}
    \end{align}
    In the second equality, all negative powers are understood as generalized inverses,
that is, they are taken only on the support of the corresponding operators.
By Lemma~\ref{lem: support lemma}, the assumption $\supp(\rho_{AE}) \subseteq\supp(I_A \otimes \sigma_E)$ implies that $\supp(\rho_{AE}) \subseteq \supp(I_A\otimes \rho_E) \subseteq \supp(I_A\otimes \sigma_E)$. 
Hence, the projector
\begin{align}
    (\pi_X \otimes \sigma_E)^{\frac{1-2\alpha}{2\alpha}}
(\pi_X \otimes \sigma_E)^{-\frac{1-2\alpha}{2\alpha}}
\end{align}
act as the identity on the support of the operator inside the norm, and therefore
do not affect its value.
In the last line we used that
\begin{equation}
     \|(\pi_X \otimes \sigma_E)^{-\frac{1-2\alpha}{2\alpha}}\|_{2\beta} = (\Tr(\pi_X \otimes \sigma_E))^\frac{1}{2\beta} =1
\end{equation}
Since the inequality holds for every $\sigma_E$ satisfying the required support
condition, and since in the definition of
$\widetilde{H}_{\alpha}^{\lambda(\alpha)}(X|E)_\rho$
the optimization may be restricted to such states---as the objective function
is equal to $+\infty$ otherwise—we may take the infimum over all admissible
$\sigma_E$ and use that
\begin{equation}
     \inf_{\sigma_E}\frac{2\alpha}{1-\alpha}\log\|\rho_{XE}^\frac{1}{2} (\pi_X \otimes \rho_E)^\frac{1}{2}(\pi_X \otimes \sigma_E)^\frac{1-2\alpha}{2\alpha}\|_{2\alpha} = \widetilde{H}_{\alpha}^{\lambda(\alpha)}(X|E)_\rho - \log{|\mathcal{X}|} \,.
\end{equation}
\end{proof}

\subsection{Proof of converse}

The above one-shot lower bound, combined with the coarse-graining inequality, implies the following converse bound for privacy amplification with quantum side information.
\begin{proposition}
\label{prop: converse}
Let $\rho_{XE}$ be a classical-quantum state and $R\geq0$. It holds that
    \begin{align}
         E_{\pa}(\rho_{XE},R) \geq \max_{\alpha \in [\frac{1}{2},1]}\frac{1-\alpha}{\alpha} \left( R-\widetilde{H}_{\alpha}^{\lambda(\alpha)}(X|E)_\rho \right) \,,
    \end{align}
    where $\lambda(\alpha) = (1-2\alpha)/(1-\alpha)$.
\end{proposition}
\begin{proof}
Let  $\{h_n~: \mc{X}^{\times n} \rightarrow \mc{Z}_n\}_{n \in \mathbb{N}}$ be an arbitrary sequence of hash functions such that
\begin{align}
\label{eq: lower rate converse}
    \liminf_{n\rightarrow \infty} \frac{1}{n}\log|\Z_n| \geq R \,.
\end{align}

For $\alpha \in [1/2,1)$,  Lemma~\ref{lem: one-shot converse} implies that
\begin{align}
\label{eq: first step converse asymptotics}
    -\frac{\alpha}{1-\alpha}\log F\left(\mathcal{R}_{h_n}(\rho_{XE}^{\otimes n}),\frac{I_{\mathcal{Z}_n}}{|\mathcal{Z}_n|}\otimes \rho_E^{\otimes n}\right) &\geq \log{|\mathcal{Z}_n|}-\widetilde{H}_{\alpha}^{\lambda(\alpha)}(Z_n|E^n)_{\mathcal{R}_n(\rho^{\otimes n})}\\
    & \geq \log{|\mathcal{Z}_n|}-\widetilde{H}_{\alpha}^{\lambda(\alpha)}(X^n|E^n)_{\rho^{\otimes n}} \\
    &  =\log{|\mathcal{Z}_n|}-n\widetilde{H}_{\alpha}^{\lambda(\alpha)}(X|E)_{\rho} \,,
\end{align}
where the inequality follows from the fact that hashing does not increase the club-sandwiched conditional entropy (see Lemma~\ref{lem: Hashing}), and the final equality follows from the additivity of the club-sandwiched conditional entropy~\cite[Theorem~4.1]{rubboli2024quantum}.
We then multiply by $1/n$ and take the $\limsup$ for $n \to \infty$, note that $\limsup \geq \liminf$, and apply equation~\eqref{eq: lower rate converse}.
Finally, we take the supremum over $\alpha \in [1/2,1)$. Since the objective function is continuous on $[1/2,1]$, this supremum is attained and can be replaced by a maximum.
\end{proof}

\section{Achievability}
\label{sec:achievability}
The proof of achievability follows standard techniques that first appeared in the context of classical channel coding in~\cite{dueck1979reliability} and~\cite[Section 2, Problem 6]{csiszar2011information}. The method was subsequently refined by Mosonyi and Ogawa~\cite{mosonyi2017strong} and later applied to the analysis of strong converse exponents for a variety of information-theoretic tasks~\cite{li2023strong,li2024operational,berta2024strong,berta2025strong,cheng2020non}.

The approach is based on the Log–Euclidean formalism combined with a double–blocking argument.
More precisely, one first derives an achievability bound for the exponent expressed in terms of Log–Euclidean quantities, which coincide with the sandwiched ones whenever all operators commute. This bound has the same form as the final error exponent, with the sandwiched quantities replaced by their Log–Euclidean counterparts. The reason is that Log–Euclidean quantities are particularly useful as they admit variational representations in terms of expressions involving the Umegaki relative entropy, thereby allowing one to invoke known achievability results for the latter.
The final step consists of a double–blocking argument, which allows one to replace the Log–Euclidean quantities with the sandwiched ones. One first performs a pinching over blocks of copies, which can be shown not to affect the error exponent and can therefore be inserted without loss. One then lets the block size tend to infinity. Since the pinching procedure introduces only logarithmic correction terms, these vanish in the asymptotic limit. Moreover, this allows one to replace the Log–Euclidean quantities by the corresponding sandwiched ones, using the fact that the sandwiched quantities arise as the infinite–block limit of the pinched Log–Euclidean expressions.

While our proof closely resembles the aforementioned arguments, our main technical contribution goes beyond a direct adaptation of the existing proofs. In particular, the expressions for the club-sandwiched conditional entropies involve three non-commuting operators, namely $\rho_{AE}$, $\rho_E$, and the state $\sigma_E$, rather than just two as for previously studied settings. As a consequence, it is not a priori clear how to implement a pinching argument, since one cannot simply pinch one state with respect to another.
Our key contribution is to introduce a double pinching procedure, consisting of the composition of a pinching with respect to the marginal state $\rho_E$ and a pinching with respect to the universal state, which can be asymptotically inserted and used to replace the optimization over $\sigma_E$. In this way, all three operators are rendered mutually commuting, and the sandwiched conditional entropies emerge as the asymptotic limit of the corresponding Log--Euclidean quantities.

\subsection{Variational expression for the LE conditional entropy}
\label{sec:variational LE}
In this subsection, we provide a variational expression for the LE conditional entropy in terms of the Umegaki relative entropy. 
This result is obtained by an argument similar to the one used to express the LE relative entropy in terms of the Umegaki relative entropy~\cite{mosonyi2017strong}.

\begin{lemma}[Variational expression LE conditional entropy]
\label{lem: variational LE}
Let $\alpha \in (0,1)$, $\lambda \leq 0$ and $\rho_{AE}$ be a state. Then, we have
\begin{equation}
(\alpha-1)H_{\alpha,\infty}^{\lambda}(A|E)_\rho =  \min_{\tau_{AE}} \{\alpha D(\tau_{AE}\|\rho_{AE})+(1-\lambda)(1-\alpha) D(\tau_E\|\rho_E)+(\alpha-1)H(A|E)_{\tau}\} \,,
\end{equation}
where we minimize over all states $\tau_{AE}$.
\end{lemma}
\begin{proof}
The case $\lambda=0$ is simpler and follows directly from the variational characterization of the Log–Euclidean relative entropy in~\eqref{eq: variational form LE}. Explicitly, 
\begin{align}
   (\alpha-1)H_{\alpha,\infty}^{\lambda=0}(A|E)_\rho & = (1-\alpha) D_{\alpha,\infty}(\rho_{AE}\|I_A \otimes \rho_E) \\
   &= \min_{\tau_{AE}}\{\alpha D(\tau_{AE}\|\rho_{AE})+(1-\alpha)D(\tau_{AE}\|I_A \otimes \rho_E)\}\\
   &= \min_{\tau_{AE}}\{\alpha D(\tau_{AE}\|\rho_{AE})+(1-\alpha)D(\tau_{E}\| \rho_E) + (\alpha-1)H(A|E)_{\tau})\} \,.
\end{align}

Let us now turn to the case $\lambda<0$. We use the (constrained) Gibbs variational principle in Lemma~\ref{lem: Gibbs principle} with projector $R_{AE}$ and Hamiltonian 
\begin{align}
    H= -\alpha\log{\rho_{AE}}-(1-\lambda)(1-\alpha)\log{\rho_E}-\lambda(1-\alpha)\log{\sigma_E} \,.
\end{align}
We obtain
\begin{align}
    &(\alpha-1)H_{\alpha,\infty}^{\lambda}(A|E) \\
    &\qquad = \sup_{\sigma_E}
 -\log{\Tr[R_{AE}\exp(\alpha\log{\rho_{AE}+(1-\alpha)R_{AE}\big((1-\lambda)\log{\rho_E}+\lambda\log{\sigma_E}\big)R_{AE}})]} \\
    &\qquad = \sup_{\sigma_E}
 \inf_{\tau_{AE}}\Big\{\Tr[\tau_{AE}\log{\tau_{AE}}] \\
    & \qquad\qquad \qquad \qquad   - \Tr[\tau_{AE}\big(\alpha\log{\rho_{AE}}+(1-\alpha)((1-\lambda)\log{\rho_E}+\lambda\log{\sigma_E})\big)]\Big\}\,,
\end{align}
where the minimization is taken over all states \(\tau_{AE}\) whose support is contained
in that of \(\rho_{AE}\). 
We then observe that, in the definition of the LE conditional entropy,
the optimization over all states \(I_A \otimes \sigma_E\) may be restricted to those whose support contains the support of $\rho_{AE}$.
Indeed, any state that violates this condition yields an objective value equal to
\(+\infty\) and therefore cannot be optimal. Hence, the optimizations can be restricted to the sets 
\begin{align}
    \{\tau_{AE}: \supp(\tau_{AE}) \subseteq \supp(\rho_{AE})\} \,\quad \{\sigma_{E}: \supp(\rho_{AE}) \subseteq \supp(I_A \otimes \sigma_{E})\}
\end{align}
The first set is compact and the second is convex. The objective
function is finite on these sets, continuous in $\tau_{AE}$, and
upper semicontinuous in $\sigma_E$, the latter following from the
lower semicontinuity of the quantum relative entropy. The above objective function is convex in $\tau_{AE}$ since the von Neumann entropy is concave in its argument and the remaining term is linear in $\tau_{AE}$. Moreover, the objective function is concave in $\sigma_{E}$ since the logarithm is operator concave and the trace is linear. We can then use 
 Sion's minimax theorem~\cite{sion1958general} to exchange the supremum with the infimum.
Next, we use the known fact that
\begin{align}
\sup_{\sigma_E}\Tr[\tau_{AE}\log{\sigma_E}] = \Tr[\tau_{AE}\log{\tau_E}] \,.
\end{align}
Here, the supremum is taken over all states that satisfy the required support condition. Nonetheless, the equality still follows since $\supp(\tau_{AE}) \subseteq \supp(I_A \otimes \tau_E)$.
We obtain
\begin{align}
    (\alpha-1)H_{\alpha,\infty}^{\lambda}(A|E) &= \inf_{\tau_{AE}}\Big\{\Tr[\tau_{AE}\log{\tau_{AE}}] \\
    & \quad  - \Tr[\tau_{AE}\big(\alpha\log{\rho_{AE}}+(1-\alpha)((1-\lambda)\log{\rho_E}+\lambda\log{\tau_E})\big)]\Big\}\\
    \label{eq: clear continuity of infimum}
    &= \inf_{\tau_{AE}} \{(1-\lambda)(1-\alpha) D(\tau_E\|\rho_E)+\alpha D(\tau_{AE}\|\rho_{AE})+(\alpha-1)H(A|E)_{\tau}\} \,,
\end{align}
 The second equality then follows from a straightforward calculation. 
From this variational expression, it is clear that the objective function is continuous in 
$\tau_{AE}$ on the set of states satisfying the required support condition. Consequently, the infimum may be replaced by a minimum, since a continuous function on a compact set always attains its minimum.

Finally, the support restriction on the states \(\tau_{AE}\)
may be removed from the infimum.
Indeed, any state \(\tau_{AE}\) that violates this condition yields an infinite
value of the relative entropy \(D(\tau_{AE}\|\rho_{AE})\) and therefore cannot
attain the infimum.

\end{proof}

\subsection{A LE conditional entropy bound}
In this subsection, we combine the previous variational formula with a standard achievability argument to derive an upper bound to the strong converse exponent in terms of the LE conditional entropy. This method is rather standard in the literature~\cite{li2025two,li2024operational}. In particular, the following proof is a straightforward adaptation of the proof provided in~\cite{li2025two} for the classical case.

We begin by deriving a variational expression for the final form of the strong converse exponent that we aim to prove, in which the club-sandwiched conditional entropy is replaced by the LE conditional entropy. We define
\begin{align}
\label{eq: definition G}
      G(\rho_{XE},R) := \max_{\alpha \in [\frac{1}{2},1]}\frac{1-\alpha}{\alpha} \left( R-H_{\alpha,\infty}^{\lambda(\alpha)}(X|E)_\rho \right) \,,
\end{align}
 where $\lambda(\alpha) = (1-2\alpha)/(1-\alpha)$.
The following lemma exploits the variational characterization of the LE conditional entropy established in Lemma~\ref{lem: variational LE} to derive a variational expression for the above quantity in terms of the Umegaki relative entropy.

\begin{lemma}
\label{lem: variational G}
Let $\rho_{XE}$ be a classical quantum state and $R\geq 0$. We have
    \begin{align}
        G(\rho_{XE},R)= \min_{\tau_{XE}} \left\{  D(\tau_E\|\rho_E)+ D(\tau_{XE}\|\rho_{XE})+|R-H(X|E)_{\tau}|^+)\right\} \,.
    \end{align}
\end{lemma}
\begin{proof}
We make use of the variational characterization established in Lemma~\ref{lem: variational LE}, specialized to the choice \(\lambda = (1-2\alpha)/(1-\alpha)\). 
We have
\begin{align}
    &\max_{\alpha \in [\frac{1}{2},1]}\frac{1-\alpha}{\alpha} \left( R-H_{\alpha,\infty}^{\lambda(\alpha)}(X|E)_\rho \right)\\
    &\qquad  = \max_{\alpha \in [\frac{1}{2},1]}\min_{\tau_{XE}} \left\{\frac{1-\alpha}{\alpha}R + D(\tau_E\|\rho_E)+ D(\tau_{XE}\|\rho_{XE})+\frac{\alpha-1}{\alpha}H(X|E)_{\tau}\right\} \\
    &\qquad  = \max_{\mu \in [0,1]}\min_{\tau_{XE}} \left\{  D(\tau_E\|\rho_E)+ D(\tau_{XE}\|\rho_{XE})+\mu(R-H(X|E)_{\tau})\right\} \\
    &\qquad  = \min_{\tau_{XE}} \max_{\mu \in [0,1]}\left\{  D(\tau_E\|\rho_E)+ D(\tau_{XE}\|\rho_{XE})+\mu(R-H(X|E)_{\tau})\right\}\\
    &\qquad  = \min_{\tau_{XE}} \left\{  D(\tau_E\|\rho_E)+ D(\tau_{XE}\|\rho_{XE})+|R-H(X|E)_{\tau}|^+)\right\}
\end{align}
In the second equality, we used that for a classical-quantum state \(\rho_{XE}\), the optimization over all quantum states can be restricted to classical-quantum states $\tau_{XE}$. 
Indeed, by the data-processing inequality for the Umegaki relative entropy, pinching a candidate state with respect to the eigenbasis of the classical register \(X\) does not increase the value of the objective function. 
Consequently, such a pinched state yields an equal or smaller value and is therefore sufficient to achieve the optimum in the variational expression. In the third equality, we set $\mu=(1-\alpha)/\alpha$ and used the Sion's minimax theorem~\cite{sion1958general} to exchange the maximum over $\mu$ and minimum over $\tau_{XE}$. Indeed, the objective function is linear (and hence concave) and continuous in \(\mu\), and convex in \(\tau_{XE}\). The latter follows from the joint convexity of the Umegaki relative entropy together with the concavity of the von Neumann conditional entropy. Moreover, the optimization over \(\tau_{XE}\) can be equivalently restricted to the compact set of states whose support is contained in the support of \(\rho_{XE}\). On this restricted domain, the objective function is finite and continuous in \(\tau_{XE}\). The final equality follows from an explicit optimization over \(\mu\).
\end{proof}

Next, we write $G(\rho_{XE},R)$ as the following minimization 
\begin{align}
        G(\rho_{XE},R)= \min \left\{  G_1(\rho_{XE},R),G_2(\rho_{XE},R)\right\} \,.
    \end{align}
    Here,
\begin{align}
        &G_1(\rho_{XE},R):= \inf_{\tau_{XE} \in \mathcal{F}_1} \left\{  D(\tau_E\|\rho_E)+ D(\tau_{XE}\|\rho_{XE})\right\} \\
        & G_2(\rho_{XE},R):= \inf_{\tau_{XE} \in \mathcal{F}_2} \left\{  D(\tau_E\|\rho_E)+ D(\tau_{XE}\|\rho_{XE})+R-H(X|E)_{\tau})\right\}
    \end{align}
where 
\begin{align}
    &\mathcal{F}_1:=\{\tau_{XE}: R < H(X|E)_{\tau}\}\,,\quad 
    \mathcal{F}_2:=\{\tau_{XE}: R \geq H(X|E)_{\tau}\}
\end{align}
We then show that the strong converse exponent is upper bounded by each of these two functions, and hence also by their minimum. 
We start with the first case.

\begin{lemma}
\label{lem: bound G1}
Let $\rho_{XE}$ be a classical-quantum state and $R\geq0$. We have
    \begin{align}
        E_{\pa}(\rho_{XE},R) \leq G_1(\rho_{XE},R)
    \end{align}
\end{lemma}
\begin{proof}
    According to the definition of $G_1(\rho_{XE},R)$, for any $\varepsilon > 0$, there exists a classical-quantum state $\tau_{XE}$ such that
    \begin{align}
        &R <H(X|E)_\tau \\
        \label{eq: second constraint G1}
        &  D(\tau_E\|\rho_E)+ D(\tau_{XE}\|\rho_{XE}) < G_1(\rho_{XE},R)+\varepsilon
    \end{align}
The result~\cite[Theorem 1]{hayashi2015precise}
 implies that for any rate $R < H(X|E)_{\tau}$ 
 there exists a sequence of hash functions $\{h_n: \mc{X}^{\times n} \rightarrow \mc{Z}_n=\{1,\ldots,\lceil e^{nR}\rceil\}\}_{n \in \mathbb{N}}$ such that\footnote{Theorem 1 of~\cite{hayashi2015precise} is formulated in terms of an expectation over a family of deterministic hash functions, with the stochasticity captured by an auxiliary random variable determining which function is applied. However, as observed later in equation (28), since the expectation is smaller than the largest over all realizations, there must exist a deterministic hash function that attains the same upper bound. The fact that Theorem 1 implies that the relative entropy tends to zero for rates below the von Neumann conditional entropy is discussed after equation (29).}
\begin{equation}
\label{equ:f3}
\lim_{n \rightarrow \infty}D\left(\mc{R}_{h_n}(\tau_{XE}^{\otimes n})\middle\| \frac{I_{Z_n}}{|\mc{Z}_n|}  \otimes \tau_E^{\otimes n}\right)=0.
\end{equation}
Next, we use that the Log-Euclidean relative entropy is larger than the sandwiched relative entropy. This is a consequence of the Araki-Lieb-Thirring trace inequality (see e.g.,~\cite[Proposition 6]{lin2015investigating}). Hence, we have that
\begin{align}
    -\log F\left(\mathcal{R}_{h_n}(\rho_{XE}^{\otimes n}),\frac{I_{\mathcal{Z}_n}}{|\mathcal{Z}_n|}\otimes \rho_E^{\otimes n}\right) \leq D_{1/2,\infty}\left(\mathcal{R}_{h_n}(\rho_{XE}^{\otimes n})\middle\|\frac{I_{\mathcal{Z}_n}}{|\mathcal{Z}_n|}\otimes \rho_E^{\otimes n}\right) \,.
\end{align}
We then use the variational form of the Log-Euclidean relative entropy in~\eqref{eq: variational form LE} and
 upper bound the minimization over all states with the specific choice $\mathcal{R}_{h_n}(\tau_{XE}^{\otimes n})$. We have
\begin{align}
    &-\log F\left(\mathcal{R}_{h_n}(\rho_{XE}^{\otimes n}),\frac{I_{\mathcal{Z}_n}}{|\mathcal{Z}_n|}\otimes \rho_E^{\otimes n}\right)  \\
    \label{eq: bound fidelity with relative entropies}
    &\qquad \leq D\left(\mathcal{R}_{h_n}(\tau_{XE}^{\otimes n})\middle\|\frac{I_{\mathcal{Z}_n}}{|\mathcal{Z}_n|}\otimes \rho_E^{\otimes n}\right)+D\left(\mathcal{R}_{h_n}(\tau_{XE}^{\otimes n})\middle\|\mathcal{R}_{h_n}(\rho_{XE}^{\otimes n})\right) \\
    & \qquad = D\left(\mathcal{R}_{h_n}(\tau_{XE}^{\otimes n})\middle\|\frac{I_{\mathcal{Z}_n}}{|\mathcal{Z}_n|}\otimes \tau_E^{\otimes n}\right)+D(\tau_{E}^{\otimes n}\|\rho_{E}^{\otimes n})+D\left(\mathcal{R}_{h_n}(\tau_{XE}^{\otimes n})\|\mathcal{R}_{h_n}(\rho_{XE}^{\otimes n})\right)\,.
\end{align}
The last equality follows by expanding the relative entropy terms and observing that
\(\Tr_{Z_n}\!\big[\mathcal{R}_{h_n}(\tau_{XE}^{\otimes n})\big]=\tau_E^{\otimes n}\), since the hashing map acts only on the register \(X^n\).

We then apply the data-processing inequality to the final term and use the additivity of the relative entropy for the second and third term to obtain
\begin{align}
    -\log F\left(\mathcal{R}_{h_n}(\rho_{XE}^{\otimes n}),\frac{I_{\mathcal{Z}_n}}{|\mathcal{Z}_n|}\otimes \rho_E^{\otimes n}\right) \leq D\left(\mathcal{R}_{h_n}(\tau_{XE}^{\otimes n})\middle\|\frac{I_{\mathcal{Z}_n}}{|\mathcal{Z}_n|}\otimes \tau_E^{\otimes n}\right)+nD(\tau_{E} \|\rho_{E})+nD\left(\tau_{XE}\|\rho_{XE}\right) \,.
\end{align}
We then divide by $n$ and take the limit  on both sides to obtain
\begin{align}
\label{eq: eq for later}
    E_{\pa}(\rho_{XE},R) \leq \limsup_{n\rightarrow+\infty}-\frac{1}{n}\log F\left(\mathcal{R}_{h'_n}(\rho_{XE}^{\otimes n}),\frac{I_{\mathcal{Z}_n}}{|\mathcal{Z}_n|}\otimes \rho_E^{\otimes n}\right) \leq D(\tau_{E}\|\rho_{E})+D\left(\tau_{XE}\|\rho_{XE}\right) \,.
\end{align}
Here, we used the fact that the first term vanishes because of equation~\eqref{equ:f3}. We also notice that
\begin{align}
    \lim_{n\rightarrow\infty}\frac{1}{n}\log{|\Z_n|} = R \,.
\end{align}
By the definition of the strong converse exponent together with the constraint~\eqref{eq: second constraint G1}, it follows that
\begin{align}
     E_{\pa}(\rho_{XE},R) \leq G_1(\rho_{XE},R)+\varepsilon
\end{align}
Letting $\varepsilon\rightarrow 0$ yields the result.

\end{proof}
Next, we prove that the strong converse exponent of privacy amplification is upper bounded by $G_2$.
\begin{lemma}
\label{lem: bound G2}
Let $\rho_{XE}$ be a classical-quantum state and $R\geq 0$. We have
    \begin{align}
        E_{\pa}(\rho_{XE},R) \leq G_2(\rho_{XE},R) \,.
    \end{align}
\end{lemma}
\begin{proof}
    By definition of \(G_2(\rho_{XE},R)\), for every \(\varepsilon>0\) there exists a classical-quantum state \(\tau_{XE}\) such that
    \begin{align}
        &H(X|E)_\tau \leq R \\
        \label{eq: second constraint G2}
        &  G_2(\rho_{XE},R) =D(\tau_E\|\rho_E)+ D(\tau_{XE}\|\rho_{XE}) + R-H(X|E)_{\tau} 
    \end{align}
    For any $\varepsilon>0$, let $R'= H(X|E)_\tau-\varepsilon$. Since, $R'<  H(X|E)_\tau$, Equation~\eqref{eq: eq for later} in the proof of the Lemma~\ref{lem: bound G1} 
implies that there exists a sequence of hash functions $\{h'_n: \mc{X}^{\times n} \rightarrow \Z'_n=\{1,\ldots,\lceil e^{nR'}\rceil\}\}_{n \in \mathbb{N}}$ such that 
\begin{equation}
\label{equ:from other}
\limsup_{n\rightarrow+\infty}-\frac{1}{n}\log F\left(\mathcal{R}_{h'_n}(\rho_{XE}^{\otimes n}),\frac{I_{\Z'_n}}{|\Z'_n|}\otimes \rho_E^{\otimes n}\right) \leq D(\tau_{E}\|\rho_{E})+D\left(\tau_{XE}\|\rho_{XE}\right) 
\end{equation}
Note that by definition, it holds that $R'< R$. We then construct a new sequence of hash functions  $\{h_n~: \mc{X}^{\times n} \rightarrow \mc{Z}_n=\{1,\ldots,\lceil e^{nR}\rceil\}\}_{n \in \mathbb{N}}$ by embedding into a larger space the above sequence of hash functions $\{\mc{R}_{h'_n}\}_{n \in \mathbb{N}}$. We then have
\begin{align}
    &-\log F\left(\mathcal{R}_{h_n}(\rho_{XE}^{\otimes n}),\frac{I_{\mathcal{Z}_n}}{|\mathcal{Z}_n|}\otimes \rho_E^{\otimes n}\right) \\
    & \qquad \qquad =-\log F\left(\mathcal{R}_{h'_n}(\rho_{XE}^{\otimes n}),\frac{I_{\Z'_n}}{|\mathcal{Z}_n|}\otimes \rho_E^{\otimes n}\right) \\
    & \qquad \qquad =-\log F\left(\mathcal{R}_{h'_n}(\rho_{XE}^{\otimes n}),\frac{I_{\Z'_n}}{|\mathcal{Z}'_n|} \frac{|\mathcal{Z}'_n|}{|\mathcal{Z}_n|} \otimes \rho_E^{\otimes n}\right) \\
    & \qquad \qquad =-\log F\left(\mathcal{R}_{h'_n}(\rho_{XE}^{\otimes n}),\frac{I_{\Z'_n}}{|\Z'_n|}\otimes \rho_E^{\otimes n}\right) + \log\frac{|\mathcal{Z}_n|}{|\mathcal{Z}'_n|}
\end{align}
We then multiply by $1/n$ and taking the limit $n\rightarrow +\infty$ on both sides. We then use equation~\eqref{equ:from other} to obtain 
\begin{align}
    E_{\pa}(\rho_{XE},R) &\leq D(\tau_{E}\|\rho_{E})+D\left(\tau_{XE}\|\rho_{XE}\right) + (R-R') \\
    & = G_2(\rho_{XE},R)+\varepsilon
\end{align}
Letting $\varepsilon \rightarrow 0$ concludes the proof.

\end{proof}

Lemma~\ref{lem: bound G1} and~\ref{lem: bound G2} prove that the strong converse exponent is smaller than both $G_1$ and $G_2$. Hence, it is also smaller than their minimum
\begin{align}
    E_{\pa}(\rho_{XE},R) \leq \min\{G_1(\rho_{XE},R),G_2(\rho_{XE},R)\} = G(\rho_{XE},R) = \max_{\alpha \in [\frac{1}{2},1]}\frac{1-\alpha}{\alpha}\left( R-H_{\alpha,\infty}^{\lambda(\alpha)}(X|E)_{\rho} \right) 
\end{align}
Hence, as a corollary, we obtain the following bound in terms of the LE conditional entropy.
\begin{corollary}
\label{cor: Bound LE}
Let $\rho_{XE}$ be a classical-quantum state and $R\geq 0$. We have
    \begin{align}
         E_{\pa}(\rho_{XE},R) \leq \max_{\alpha \in [\frac{1}{2},1]}\frac{1-\alpha}{\alpha} \left( R-H_{\alpha,\infty}^{\lambda(\alpha)}(X|E)_{\rho} \right)\,,
    \end{align}
     where $\lambda(\alpha) = (1-2\alpha)/(1-\alpha)$.
\end{corollary}

\subsection{A club-sandwiched conditional entropy bound}

The final step is to derive an analogous upper bound as above in terms of the club-sandwiched conditional entropy. The proof is similar to that of~\cite{li2024operational} and leverages a pinching argument.

The first step is to show that the strong converse exponent is essentially unaffected by pinching the state with respect to an operator whose spectral cardinality grows only polynomially with the number of copies. This property follows from the fact that the fidelity belongs to the family of sandwiched relative entropies. For this class of quantities, it is known that an additional pinching operation can be asymptotically inserted without affecting their value; see, for instance,~\cite[Proposition~4.12]{tomamichel16_book}. Since we will later require the pinched $\rho_{XE}^{\otimes m}$ to commute with both $\rho_E^{\otimes m}$ and $\omega_{E^m}$, we consider the composition of the corresponding pinching maps $\mathcal P_{\rho_E^{\otimes m}}\circ \mathcal P_{\omega_{E^m}}(\rho_{XE}^{\otimes m})$.

\begin{lemma}
\label{lem: LE bound with pinched state}
Let $\rho_{XE}$ be a classical-quantum state and $R\geq 0$. Moreover, let  $\tilde{\rho}_{X^mE^m}=\P_{\rho_E^{\otimes m}} \circ\P_{\omega_{E^m}}(\rho_{XE}^{\otimes m})$. For any $m\geq 1$, it holds that
    \begin{align}
    E_{\pa}(\rho_{XE},R) \leq \max_{\alpha \in [\frac{1}{2},1]}\frac{1-\alpha}{\alpha}\left(R-\frac{1}{m}H_{\alpha,\infty}^{\lambda(\alpha)}(X^m|E^m)_{\widetilde{\rho}}\right) + \frac{2}{m} (d_E-1)\log{(m+1)} \,,
\end{align}
 where $\lambda(\alpha) = (1-2\alpha)/(1-\alpha)$.
\end{lemma}
\begin{proof}
We fix $m\geq 1$. The proof of Corollary~\ref{cor: Bound LE} applied to this state implies that there exists a sequence of hash functions $\{h_{m,k}: \X^{\times mk} \rightarrow Z_{m,k}\}_{k\in\mathbb{N}}$ such that 
\begin{align}
\label{eq: hashing m,k}
    \liminf_{k\rightarrow \infty}\frac{1}{k}\log{|\Z_{m,k}|}\geq mR
\end{align}
and 
\begin{align}
\label{eq: hashing previous lemma}
\limsup_{k\rightarrow\infty}-\frac{1}{k}\log F\left(\R_{h_{m,k}}(\widetilde{\rho}_{X^mE^m}^{\otimes k}),\frac{I_{\Z_{m,k}}}{|\Z_{m,k}|}\otimes \rho_E^{\otimes mk}\right) \leq \max_{\alpha \in\left [\frac{1}{2},1\right]}\frac{1-\alpha}{\alpha}\{mR-H_{\alpha,\infty}^{\lambda(\alpha)}(X^m|E^m)_{\widetilde{\rho}_{X^mE^m}}\}
    \end{align}
    Here, in the second argument of the fidelity, we used that the marginal on $E^m$ is invariant under the applied pinching. Indeed, since the hash function acts only on the first marginal, and since the universal state $\omega_{E^m}$ commutes with $\rho_E^{\otimes m}$, Lemma~\ref{lem: marginal of commuting pinching} implies that
    \begin{align}
        \Tr_{Z_{m,k}}[\R_{h_{m,k}}(\widetilde{\rho}_{X^mE^m}^{\otimes k})] = \big(\mathcal{P}_{\rho_E^{\otimes m}} \circ \mathcal{P}_{\omega_E^{\otimes m}}(\rho_E^{\otimes m})\big)^{\otimes k} = \rho_E^{\otimes mk} \,.
    \end{align}
Next, for each such \(n\), we construct a hashing map \(\mathcal{R}_{h_n}\) from \(\mathcal{R}_{h_{m,k}}\) as follows. 
We first determine the integer \(k\) such that \(n \in [km,(k+1)m)\), and choose
\(\mathcal Z_n = \mathcal Z_{m,k}\). To emphasize the dependence of $k$ on $n$, we write $k = k(n)$ whenever needed.
Here, note that, for all \(n \in [km,(k+1)m)\), the value of $k$ remains fixed. 
We then define \(h_n\) by
\begin{align}
    h_n(x_1,\ldots,x_{mk},x_{mk+1},\ldots,x_n)
:=  h_{m,k}(x_1,\ldots,x_{mk}).
\end{align}
It is straightforward to verify that
\begin{align}
    \mathcal  \R_{h_n}\bigl(\rho_{XE}^{\otimes n}\bigr)
=
\mathcal  \R_{h_{m,k}}\bigl(\rho_{XE}^{\otimes mk}\bigr)
\otimes
\rho_E^{\otimes (n-mk)}.
\end{align}
Moreover, since the choice of $k$ depends on $n$, and we can write $n = mk+\ell$, with $0\leq \ell \leq m-1$, we have that
\begin{align}
    \liminf_{n\rightarrow \infty}\frac{1}{n}\log{|\Z_{n}|} &\geq  \liminf_{n\rightarrow \infty}\frac{1}{mk(n)+\ell(n)}\log{|\Z_{m,k(n)}|} \\
    &\geq  \liminf_{n\rightarrow \infty}\frac{1}{mk(n)+m-1}\log{|\Z_{m,k(n)}|}\\
    &=\liminf_{k\rightarrow \infty} \frac{mk}{mk+m-1}\frac{1}{mk}\log{|\Z_{m,k}|} \\
    & = \frac{1}{m} \liminf_{k\rightarrow \infty} \frac{1}{k}\log{|\Z_{m,k}|}\\
    &\geq R\,.
\end{align}
Here, in order to rewrite the limit in $n$ as a limit in $k$, we used that repeating each term a fixed number $m$ of times leaves the $\liminf$ unchanged. In addition, in the last equality, we used that $\lim_{k\to\infty}\frac{mk}{mk+m-1}=1$.
Finally, the last inequality follows from Eq.~\eqref{eq: hashing m,k}. Hence, the hash functions $\mathcal R_n$ achieve extraction rate $R$, and therefore can be used to upper bound the strong converse exponent.

We then insert the pinching operators. We  apply twice the inequality in Lemma~\ref{lem: pinching lemma} with $\ell=0$ together with
\begin{align}
    & \max\{|\mathrm{spec}(\rho_E^{\otimes m})|, |\mathrm{spec}(\omega_{E^{m}})|\} \leq (m+1)^{d_E-1} \,.
\end{align}
To apply this lemma, we must check that $\rho_E^{\otimes m}$ is block-diagonal with respect to the joint eigenspaces of $\rho_E^{\otimes m}$ and $\omega_{E^m}$. This follows from the fact that $\rho_E^{\otimes m}$ commutes with $\omega_{E^m}$ and is therefore invariant under pinching with respect to either operator, i.e.,
\begin{equation}
\mathcal{P}_{\rho_E^{\otimes m}}\!\left(\rho_E^{\otimes m}\right)
=
\mathcal{P}_{\omega_{E^m}}\!\left(\rho_E^{\otimes m}\right)
=
\rho_E^{\otimes m}.
\end{equation}
We then obtain
\begin{align}
    &-\log F\left(\R_{h_{m,k}}(\widetilde{\rho}_{X^mE^m}^{\otimes k}),\frac{I_{\Z_{m,k}}}{|\Z_{m,k}|}\otimes \rho_E^{\otimes mk}\right) \\
    & \qquad\qquad = -\log F\left(\mathcal{P}^{\otimes k}_{\rho_E^{\otimes m}}(\R_{h_{m,k}}(\mathcal{P}^{\otimes k}_{\omega_{E^{m}}}(\rho_{XE}^{\otimes m k}))),\frac{I_{\Z_{m,k}}}{|\Z_{m,k}|}\otimes \rho_E^{\otimes mk}\right) \\
    & \qquad \qquad \geq -\log F\left(\R_{h_{m,k}}(\mathcal{P}^{\otimes k}_{\omega_{E^{m}}}(\rho_{XE}^{\otimes m k})),\frac{I_{\Z_{m,k}}}{|\Z_{m,k}|}\otimes \rho_E^{\otimes mk}\right) - k(d_E-1)\log{(m+1)} \\
    & \qquad \qquad \geq -\log F\left(\R_{h_{m,k}}(\rho_{XE}^{\otimes m k})\otimes \rho_E^{\otimes (n-mk)},\frac{I_{\Z_{m,k}}}{|\Z_{m,k}|}\otimes \rho_E^{\otimes n}\right) - 2k(d_E-1)\log{(m+1)} \\
     & \qquad \qquad = -\log F\left(\R_{h_{n}}(\rho_{XE}^{\otimes n}),\frac{I_{\Z_{m,k}}}{|\Z_{m,k}|}\otimes \rho_E^{\otimes n}\right) - 2k(d_E-1)\log{(m+1)}
\end{align}
Using this inequality, the fact that $n \geq mk$ together with the positivity of the negative logarithm of the fidelity between states, and inequality~\eqref{eq: hashing previous lemma}, we can upper bound the strong converse exponent as
\begin{align}
    E_{\pa}(\rho_{XE},R) &\leq \limsup_{n\rightarrow\infty}-\frac{1}{n}\log F\left(\R_{h_{n}}(\rho_{XE}^{\otimes n}),\frac{I_{\Z_{m,k(n)}}}{|\Z_{m,k(n)}|}\otimes \rho_E^{\otimes n}\right) \\
    & \leq \limsup_{k\rightarrow\infty}-\frac{1}{mk}\log F\left(\R_{h_{m,k}}(\widetilde{\rho}_{X^mE^m}^{\otimes k}),\frac{I_{\Z_{m,k}}}{|\Z_{m,k}|}\otimes \rho_E^{\otimes mk}\right) + 2\frac{1}{m} (d_E-1)\log{(m+1)} \\
    & \leq \max_{\alpha \in [\frac{1}{2},1]}\frac{1-\alpha}{\alpha}\left( R-\frac{1}{m}H_{\alpha,\infty}^{\lambda(\alpha)}(X^m|E^m)_{\widetilde{\rho}} \right) + 2\frac{1}{m} (d_E-1)\log{(m+1)}
\end{align}
\end{proof}
In the next step, we leverage the pinching structure in the above expression to establish that, 
for large $m$, the LE conditional entropy approaches the 
corresponding club-sandwiched conditional entropy.

\begin{proposition}
\label{prop: bound sandwiched}
Let $\rho_{XE}$ be a classical-quantum state and $R\geq 0$. We have
    \begin{align}
         E_{\pa}(\rho_{XE},R) \leq \max_{\alpha \in \left[\frac{1}{2},1\right]}\frac{1-\alpha}{\alpha}\left( R-\widetilde{H}_{\alpha}^{\lambda(\alpha)}(X|E)_{\rho} \right)\,,
    \end{align}
     where $\lambda(\alpha) = (1-2\alpha)/(1-\alpha)$.
\end{proposition}

\begin{proof}
We have that for any permutation invariant state $\tau_{E^{ m}}$, it holds that
\begin{align}
     \tau_{E^m} \leq \omega_{E^m}(m+1)^{d^2_E-1} \,.
\end{align}
We denote with $R_{X^m E^m}$ the projector onto the support of $\tilde{\rho}_{X^mE^m} = \mathcal{P}_{\rho_E^{\otimes m}}\circ \mathcal{P}_{\omega_{E^m}}(\rho_{XE}^{\otimes m})$. We know by Lemma~\ref{lem: permutation invariant optimum} that we can choose the optimizer $\tau_{E^m}$ of the LE conditional entropy of $\tilde{\rho}_{X^mE^m}$ to be permutation invariant. By Lemma~\ref{lem: marginal of commuting pinching}, we also have that $\Tr_{X^m}[\tilde{\rho}_{X^mE^m}] = \rho_E^{\otimes m}$. This implies that
    \begin{align}
&H_{\alpha,\infty}^\lambda(X^m|E^m)_{\tilde{\rho}} \notag\\
&\; =\frac{1}{1-\alpha}\log\Tr\Big[R_{X^m E^m}\exp\Big(\alpha\log \tilde{\rho}_{X^mE^m} + (1-\alpha)R_{X^m E^m}\big((1-\lambda)\log{\rho_E^{\otimes m}}+\lambda\log{\tau_{E^m}}\big)R_{X^m E^m}\Big)\Big] \\
&\; \geq \frac{1}{1-\alpha}\log\Tr\Big[R_{X^m E^m}\exp\Big(\alpha\log \tilde{\rho}_{X^mE^m}
+(1-\alpha)R_{X^m E^m}\big((1-\lambda)\log{\rho_E^{\otimes m}}+\lambda\log{\omega_{E^m}}\big)R_{X^m E^m}\Big)\Big] \notag\\
& \qquad\qquad   +\lambda (d^2_E-1)\log{(m+1)} 
\end{align}
The inequality follows from the operator monotonicity of the logarithm, together with the fact that the trace functional \(M \mapsto \Tr[f(M)]\) inherits the monotonicity of the function \(f\) (see, e.g.,~\cite{carlen2010trace}). In the present setting, this applies to the exponential function. Next, we invoke Lemma~\ref{lem: pinching with identity}, which implies that pinching with respect to 
$I_A \otimes \rho_E$ factorizes as the identity on $A$ tensored with pinching on $E$, i.e., $ \mathcal P_{I_A \otimes \rho_E} = I_A \otimes \mathcal P_{\rho_E}$.
As a result, all operators under consideration commute pairwise, i.e.,  
\begin{align}
    [\rho_E^{\otimes m},\omega_{E^m} ] =[\tilde{\rho}_{X^mE^m}, I_{X^m} \otimes \rho_E^{\otimes m}] =  [\tilde{\rho}_{X^mE^m}, I_{X^m} \otimes \omega_{E^m}]  = 0 . 
\end{align}
Here, the first commutation relation follows from the defining properties of the universal state. 
The remaining commutation relations follow from the fact that the pinching maps can be interchanged, 
since their spectral projectors commute as the underlying operators commute. 
Hence, since all operators commute, we can rewrite the LE conditional entropy in terms of the club-sandwiched conditional entropy. We have that
\begin{align}
&H_{\alpha,\infty}^\lambda(X^m|E^m)_{\tilde{\rho}} \notag\\
&\;\geq \frac{1}{1-\alpha}\log{\Tr(\omega_{E^m}^{\frac{\lambda}{2}\frac{1-\alpha}{\alpha}}\rho_{E}^{\otimes m \frac{1-\lambda}{2}\frac{1-\alpha}{\alpha}} \tilde{\rho}_{X^mE^m} \rho_{E}^{\otimes m\frac{1-\lambda}{2}\frac{1-\alpha}{\alpha}}\omega_{E^m}^{\frac{\lambda}{2}\frac{1-\alpha}{\alpha}})^\alpha} \\
& \qquad +\lambda (d^2_E-1)\log{(m+1)}  \\
& \; \geq \inf_{\sigma_{E^m}}\frac{1}{1-\alpha}\log{\Tr(\sigma_{E^m}^{\frac{\lambda}{2}\frac{1-\alpha}{\alpha}}\rho_{E}^{\otimes m \frac{1-\lambda}{2}\frac{1-\alpha}{\alpha}} \tilde{\rho}_{X^mE^m} \rho_{E}^{\otimes m\frac{1-\lambda}{2}\frac{1-\alpha}{\alpha}}\sigma_{E^m}^{\frac{\lambda}{2}\frac{1-\alpha}{\alpha}})^\alpha}  + \lambda (d^2_E-1)\log{(m+1)} \\
& \; = \widetilde{H}_{\alpha}^\lambda(X^m|E^m)_{\tilde{\rho}} +\lambda (d^2_E-1)\log{(m+1)}\\
& \; \geq \widetilde{H}_{\alpha}^\lambda(X^m|E^m)_{\rho^{\otimes m}} +\lambda (d^2_E-1)\log{(m+1)} \\
& \; = m\widetilde{H}_{\alpha}^\lambda(X|E)_{\rho} +\lambda (d^2_E-1)\log{(m+1)} 
\end{align}
The third inequality is a consequence of the monotonicity of the conditional entropy under maps acting on the side information. The last equality is due to additivity.

Inserting this inequality in the bound for the error exponent in Lemma~\ref{lem: LE bound with pinched state} gives 
\begin{align}
    E_{\pa}(\rho_{XE},R)& \leq \max_{\alpha \in [\frac{1}{2},1]}\left\{\frac{1-\alpha}{\alpha}\left( R-\widetilde{H}_{\alpha}^{\lambda(\alpha)}(X|E)_\rho \right) + \frac{1}{m}\left(\frac{2\alpha-1}{\alpha}+2\right)(d^2_E-1)\log{(m+1)}\right\} \\
    &\leq \max_{\alpha \in [\frac{1}{2},1]}\frac{1-\alpha}{\alpha}\left( R-\widetilde{H}_{\alpha}^{\lambda(\alpha)}(X|E)_{\rho} \right) + \frac{3}{m}(d^2_E-1)\log{(m+1)} 
\end{align}
Since the last term is a constant that does not depend on $\alpha$, we can let $m$ be arbitrarily large to complete the proof.
    
\end{proof}

\section{Universal protocols}
\label{sec: universal}

As discussed in Section~\ref{sec: universal protocols}, one may consider a setting with additional public randomness that is likewise available to the eavesdropper. This enables the construction of a protocol that does not rely on the underlying state, which is typically unknown to the honest parties, especially since the eavesdropper’s marginal is not known. In this way, the added randomness renders the protocol universal, i.e., independent of the state.

Let us denote the state after the application of the hash function with seed $H_n$ to $n$ copies of the state $\rho_{XE}$ by
\begin{align}
\rho_{Z_n E_n H_n}
\;:=\;
\sum_{h_n\in \mathcal{H}_n} 
p_n(h_n)\,
\mathcal{R}_{h_n}(\rho_{XE}^{\otimes n})
\otimes 
|h_n\rangle\!\langle h_n|_{H_n}\,.
\end{align}
The corresponding ideal state is
\begin{align}
\omega_{Z_n E_n H_n}
\;:=\;
\frac{I_{\mathcal{Z}_n}}{|\mathcal{Z}_n|}
\otimes 
\rho_E^{\otimes n}
\otimes 
\rho_{H_n},
\qquad
\rho_{H_n} 
:=
\sum_{h_n\in \mathcal{H}_n} 
p_n(h_n)\,
|h_n\rangle\!\langle h_n|_{H_n},
\end{align}
which represents a uniform key independent of both $E_n$ and the public seed $H_n$. 
We define the seeded strong converse exponent as follows
\begin{align}
\label{eq: universal error exponent}
E_{\pa}^{\mathrm{s}}(\rho_{XE},R) &:=
\inf 
\Bigg\{
\limsup_{n\to\infty}
-\frac{1}{n}
\log 
F\!\left(
\rho_{Z_n E_n H_n},
\omega_{Z_n E_n H_n}
\right)
\;\Bigg|\;
\liminf_{n\to\infty}
\frac{1}{n}\log |\mathcal{Z}_n|
\ge R
\Bigg\}\\
&=
\inf 
\Bigg\{
\limsup_{n\to\infty}
-\frac{2}{n}
\log 
\mathbb{E}_{H_n \sim p_n}
F\!\left(
\mathcal{R}_{H_n}(\rho_{XE}^{\otimes n}),
\frac{I_{\mathcal{Z}_n}}{|\mathcal{Z}_n|}
\otimes 
\rho_E^{\otimes n}
\right)^{\frac{1}{2}}
\;\Bigg|\;
\liminf_{n\to\infty}
\frac{1}{n}\log |\mathcal{Z}_n|
\ge R
\Bigg\},
\end{align}
where the infimum is taken over all sequences of families of hash functions 
$\{\mathcal{H}_n\}_{n\in\mathbb{N}}$, probability distributions 
$p_n$ on $\mathcal{H}_n$, and corresponding families of channels 
$\{\mathcal{R}_{h_n} : h_n \in \mathcal{H}_n\}$ 
implementing hash functions 
$h_n : \mathcal{X}^{\times n} \to \mathcal{Z}_n$.
The second equality follows by expanding the expression in the definition, analogously to the argument leading to~\eqref{eq: fidelity with randomness}.

We show that, in this setting, the error exponent coincides with the one obtained without additional randomness. Moreover, every sequence of two-universal hash families with the prescribed asymptotic output rate achieves the optimal exponent. In particular, the two-universal hash families may be chosen independently of the state \(\rho_{XE}\), thereby yielding a universal construction.

We first recall the concept of two-universal hash functions.

\begin{definition}
Let $\mathcal{H}$ be a finite seed space and let 
$\{h\}_{h\in\mathcal{H}}$ 
be a family of hash functions
\begin{align}
h : \mathcal{X} \to \mathcal{Z},
\qquad 
\mathcal{Z} := \{1,\dots,|\mathcal{Z}|\}.
\end{align}
Let $H$ be a random variable distributed according to $p$ on $\mathcal{H}$. The family $\{h\}_{h\in\mathcal{H}}$ is called 
two-universal if for all distinct 
$x,x' \in \mathcal{X}$,
\begin{align}
\Pr_{H \sim p}\!\big[
h_H(x) = h_H(x')
\big]
\le
\frac{1}{|\mathcal{Z}|}\,.
\end{align}
Here, $h_H$ denotes the function selected by the seed $H$ from the family $\{h\}_{h\in\mathcal{H}}$.
\end{definition}
A standard choice is linear Toeplitz hashing.

\begin{theorem}
\label{thm: main error exponent s}
Let $\rho_{XE}$ be a classical-quantum state and $R\geq 0$. Then, we have
\begin{align}
    E_{\pa}^{\mathrm{s}}(\rho_{XE},R)=E_{\pa}(\rho_{XE},R) \, .
\end{align}
Moreover, every sequence of two-universal hash families satisfying $\lim_{n\to\infty}\frac{1}{n}\log|\mathcal{Z}_n|= R$
achieves this exponent. In particular, the hash families and the seed distributions may be chosen independently of the state $\rho_{XE}$.
\end{theorem}
The optimal protocol can be realized using a sequence of two-universal
families of hash functions constructed from standard linear Toeplitz
matrices, after injectively embedding $\mathcal{X}^n$ into a binary
vector space, with output alphabet
$\{0,1\}^{\lfloor nR/\log 2\rfloor}$.

\medskip

The bound $E_{\pa}^{\mathrm{s}}(\rho_{XE},R)\geq E_{\pa}(\rho_{XE},R)$ follows directly by choosing the maximizing hash function 
\begin{align}
    \mathbb{E}_{H_n}F\left(\mathcal{R}_{H_n}(\rho_{XE}^{\otimes n}),\frac{I_{\mathcal{Z}_n}}{|\mathcal{Z}_n|}\otimes\rho_E^{\otimes n}\right)^\frac{1}{2} \leq \max_{h_n \in\mathcal{H}_n}F\left(\mathcal{R}_{h_n}(\rho_{XE}^{\otimes n}),\frac{I_{\mathcal{Z}_n}}{|\mathcal{Z}_n|}\otimes\rho_E^{\otimes n}\right)^\frac{1}{2}
\end{align}
Thus, the lower bound is the same as the one in Proposition~\ref{prop: converse}, and no separate converse argument is needed. This comparison applies to any seeded protocol, without requiring two-universality.

For achievability, we follow the LE conditional entropy and double-pinching argument of Section~\ref{sec:achievability}. The additional ingredient is a bound on the average relative entropy that holds for every two-universal hash family. 
\begin{lemma}
\label{lem: universal one-shot leakage}
Let $\tau_{XE}$ be a classical-quantum state. For every $\alpha\in(1,2]$ and two-universal family of hash functions from $\X$ to $\Z$ it holds that
\begin{align}
\label{eq: universal one-shot leakage}
    \mathbb{E}_{H}D\left(\mathcal{R}_{H}(\tau_{XE})\middle\|\frac{I_\Z}{|\Z|}\otimes\tau_E\right)
    \leq\log|\mathrm{spec}(\tau_E)|+\frac{\log 2}{\alpha-1}
    +|\log|\Z|-\widetilde{H}^{\downarrow}_{\alpha}(X|E)_\tau|^+ \,.
\end{align}
\end{lemma}
\begin{proof}
Let $\widetilde{\tau}_{XE}=\mathcal{P}_{\tau_E}(\tau_{XE})$. The pinching argument following~\cite[Theorem~2]{hayashi2015precise} gives~\footnote{In the corresponding pinching argument, the arXiv
version (arXiv:1202.0601v3, Sec.~5) uses the Petz conditional
R\'enyi entropy, whereas the published version uses the
sandwiched conditional R\'enyi entropy. The two quantities
coincide on $\widetilde{\tau}_{XE}$, since
$[\widetilde{\tau}_{XE},I_X\otimes\tau_E]=0$, and therefore
give the same intermediate bound.}
\begin{align}
\label{eq: universal leakage moment}
    &\mathbb{E}_{H}\exp\left((\alpha-1)D\left(\mathcal{R}_{H}(\tau_{XE})\middle\|\frac{I_\Z}{|\Z|}\otimes\tau_E\right)\right)\notag\\
    &\qquad\leq |\mathrm{spec}(\tau_E)|^{\alpha-1}\left[1+\exp\left((\alpha-1)\left(\log|\Z|-\widebar{H}^{\downarrow}_{\alpha}(X|E)_{\widetilde{\tau}}\right)\right)\right]\\
    &\qquad\leq |\mathrm{spec}(\tau_E)|^{\alpha-1}\left[1+\exp\left((\alpha-1)\left(\log|\Z|-\widetilde{H}^{\downarrow}_{\alpha}(X|E)_\tau\right)\right)\right] \,.
\end{align}
For the second inequality, we used that the Petz and sandwiched conditional entropies coincide for $\widetilde{\tau}_{XE}$, since it commutes with $I_X\otimes\tau_E$, and that data processing implies $\widetilde{H}^{\downarrow}_{\alpha}(X|E)_{\widetilde{\tau}}\geq\widetilde{H}^{\downarrow}_{\alpha}(X|E)_\tau$. Applying Jensen's inequality to the left-hand side, taking the logarithm, and using $\log(1+e^t)\leq\log 2+|t|^+$ proves the claim.
\end{proof}

We next apply this bound to products of blocks of a fixed size. The remaining copies are kept as inputs to the hash function, so that the two-universality condition continues to hold on the full input alphabet.

\begin{lemma}
\label{lem: universal block leakage}
Let $m\geq 1$ and let $\tau_{X^mE^m}$ and $\rho_{XE}$ be classical-quantum states. For every $n$, write $n=mk+\ell$, with $0\leq\ell<m$, and define $\tau^{(m)}_{X^nE^n}:=\tau_{X^mE^m}^{\otimes k}\otimes\rho_{XE}^{\otimes\ell}$.
For every sequence of two-universal hash families satisfying $\lim_{n\to\infty}\frac{1}{n}\log|\mathcal{Z}_n|= R$, we have
\begin{align}
\label{eq: universal block leakage}
    \limsup_{n\to\infty}\frac{1}{n}\mathbb{E}_{H_n}
    D\left(\mathcal{R}_{H_n}(\tau^{(m)}_{X^nE^n})\middle\|\frac{I_{\Z_n}}{|\Z_n|}\otimes\tau^{(m)}_{E^n}\right)
    \leq\left|R-\frac{1}{m}H(X^m|E^m)_\tau\right|^+ \,.
\end{align}
\end{lemma}
\begin{proof}
By Lemma~\ref{lem: universal one-shot leakage} and the additivity of the arrow-down sandwiched conditional entropy, for every $\alpha\in(1,2]$ we have
\begin{align}
    \label{eq:block-one-shot}&\mathbb{E}_{H_n}D\left(\mathcal{R}_{H_n}(\tau^{(m)}_{X^nE^n})\middle\|\frac{I_{\Z_n}}{|\Z_n|}\otimes\tau^{(m)}_{E^n}\right)\notag\\
    &\qquad\leq\log|\mathrm{spec}(\tau^{(m)}_{E^n})|+\frac{\log 2}{\alpha-1}
    +\left|\log|\Z_n|-k\widetilde{H}^{\downarrow}_{\alpha}(X^m|E^m)_\tau
    -\ell\widetilde{H}^{\downarrow}_{\alpha}(X|E)_\rho\right|^+ \,.
\end{align}
The size of the spectrum is submultiplicative under tensor products. We then have the following upper bound on the size of the spectrum 
\begin{align}
    |\spec(\tau^{(m)}_{E^n})|  \leq |\spec(\tau^{\otimes k}_{E^m})| |\spec(\rho^{\otimes \ell}_{E})| \leq (k+1)^{d_{E^m}-1} d_{E}^\ell \,.
\end{align}
Hence, it holds that for every fixed $\ell$
\begin{equation}
  \lim_{k\to \infty}\frac1{mk+\ell}\log |\spec(\tau^{(m)}_{E^{mk+\ell}})|=0 \,.
\end{equation}
Next, we divide \eqref{eq:block-one-shot} by $mk+\ell$ and let $k\to\infty$ with fixed  $\ell$.  The constant term $\log2/(\alpha-1)$ does not contribute in the limit. Moreover, we have that $\widetilde{H}^{\downarrow}_{\alpha}(X^\ell|E^\ell)_{\rho^{\otimes \ell}} \leq \ell\log{d_{X}}$ and 
\begin{equation}
  \lim_{k\to\infty}\frac1{mk+\ell}\log|\mathcal Z_{mk+\ell}| = R,
  \qquad
 \lim_{k\to\infty} \frac{k}{mk+\ell}=\frac1m.
\end{equation}
We then use that $[a]_+/n=[a/n]_+$ and obtain
\begin{align}
&\limsup_{k\to\infty}\frac1{mk+\ell}
\E_{H_{mk+\ell}}D\!\left(
\cR_{H_{mk+\ell}}\big(\tau^{(m)}_{X^{mk+\ell}E^{mk+\ell}}\big)
\middle\|
\frac{I_{Z_{mk+\ell}}}{|\mathcal Z_{mk+\ell}|}\otimes
\tau^{(m)}_{E^{mk+\ell}}
\right)\le
\pos{R-\frac1m\widetilde H^{\downarrow}_{\alpha}(X^m|E^m)_\tau}.
\label{eq:fixed-residue}
\end{align}
The limit superior of a sequence equals the maximum of the limit superiors of its finitely many residue-class subsequences; explicitly, for any sequence \(\{a_n\}_n\) and any \(m\in\mathbb{N}\), \(\limsup_{n\to\infty} a_n=\max_{\ell\in\{0,\ldots,m-1\}}\limsup_{k\to\infty} a_{mk+\ell}\).
Since the inequality holds for every \(\alpha\in(1,2]\), taking the limit \(\alpha\to 1^+\), using the convergence of the arrow-down sandwiched conditional entropy to the standard conditional entropy and invoking the continuity of the positive-part function yields the result.
\end{proof}

We now combine the relative entropy bound with the variational expression for the LE conditional entropy. 
\begin{lemma}
\label{lem: universal LE bound}
Let $m\geq 1$ and let $\widetilde{\rho}_{X^mE^m}$ and $\rho_{XE}$ be classical-quantum states. For $n=mk+\ell$, with $0\leq\ell<m$, set $\widetilde{\rho}^{(m)}_{X^nE^n}:=\widetilde{\rho}_{X^mE^m}^{\otimes k}\otimes\rho_{XE}^{\otimes\ell}$.
For every sequence of two-universal hash families satisfying $\lim_{n\to\infty}\frac{1}{n}\log|\mathcal{Z}_n|= R$, we have
\begin{align}
\label{eq: universal LE bound}
    &\limsup_{n\to\infty}-\frac{2}{n}\log\mathbb{E}_{H_n}
    F\left(\mathcal{R}_{H_n}(\widetilde{\rho}^{(m)}_{X^nE^n}),\frac{I_{\Z_n}}{|\Z_n|}\otimes\widetilde{\rho}^{(m)}_{E^n}\right)^{\frac{1}{2}}\leq\max_{\alpha\in[\frac{1}{2},1]}\frac{1-\alpha}{\alpha}
    \left(R-\frac{1}{m}H_{\alpha,\infty}^{\lambda(\alpha)}(X^m|E^m)_{\widetilde{\rho}}\right) \, ,
\end{align}
where $\lambda(\alpha)=(1-2\alpha)/(1-\alpha)$.
\end{lemma}
\begin{proof}
Fix a classical-quantum state $\tau_{X^mE^m}$ whose support is contained in that of $\widetilde{\rho}_{X^mE^m}$, and define  $\tau^{(m)}_{X^nE^n}:=\tau_{X^mE^m}^{\otimes k}\otimes\rho_{XE}^{\otimes\ell}$. As in~\eqref{eq: bound fidelity with relative entropies}, the comparison with the LE relative entropy at order $1/2$ and its variational expression give, for every hash function $h_n$,
\begin{align}
    &-\log F\left(\mathcal{R}_{h_n}(\widetilde{\rho}^{(m)}_{X^nE^n}),\frac{I_{\Z_n}}{|\Z_n|}\otimes\widetilde{\rho}^{(m)}_{E^n}\right)\notag\\
    &\qquad\leq D\left(\mathcal{R}_{h_n}(\tau^{(m)}_{X^nE^n})\middle\|\mathcal{R}_{h_n}(\widetilde{\rho}^{(m)}_{X^nE^n})\right)
    +D\left(\mathcal{R}_{h_n}(\tau^{(m)}_{X^nE^n})\middle\|\frac{I_{\Z_n}}{|\Z_n|}\otimes\widetilde{\rho}^{(m)}_{E^n}\right)\notag\\
    &\qquad\leq kD(\tau_{X^mE^m}\|\widetilde{\rho}_{X^mE^m})
    +kD(\tau_{E^m}\|\widetilde{\rho}_{E^m})
    +D\left(\mathcal{R}_{h_n}(\tau^{(m)}_{X^nE^n})\middle\|\frac{I_{\Z_n}}{|\Z_n|}\otimes\tau^{(m)}_{E^n}\right) \,.
\end{align}
In the last inequality, we used data processing for the first relative entropy and split the second relative entropy into the leakage relative to $\tau^{(m)}_{E^n}$ and the relative entropy of the marginals. Since hashing acts only on the classical register, it leaves the marginal unchanged. The remaining $\ell$ copies are identical in the two block products and therefore do not contribute to either relative entropy between the input states.

By Jensen's inequality for the logarithm, averaging the preceding bound yields
\begin{align}
\label{eq: universal averaged trial bound}
    &-2\log\mathbb{E}_{H_n}F\left(\mathcal{R}_{H_n}(\widetilde{\rho}^{(m)}_{X^nE^n}),\frac{I_{\Z_n}}{|\Z_n|}\otimes\widetilde{\rho}^{(m)}_{E^n}\right)^{\frac{1}{2}}\notag\\
    &\qquad\leq kD(\tau_{X^mE^m}\|\widetilde{\rho}_{X^mE^m})+kD(\tau_{E^m}\|\widetilde{\rho}_{E^m})
    +\mathbb{E}_{H_n}D\left(\mathcal{R}_{H_n}(\tau^{(m)}_{X^nE^n})\middle\|\frac{I_{\Z_n}}{|\Z_n|}\otimes\tau^{(m)}_{E^n}\right) \,.
\end{align}
Dividing by $n$, taking the limit superior, and applying Lemma~\ref{lem: universal block leakage}, we obtain
\begin{align}
    &\limsup_{n\to\infty}-\frac{2}{n}\log\mathbb{E}_{H_n}
    F\left(\mathcal{R}_{H_n}(\widetilde{\rho}^{(m)}_{X^nE^n}),\frac{I_{\Z_n}}{|\Z_n|}\otimes\widetilde{\rho}^{(m)}_{E^n}\right)^{\frac{1}{2}}\notag\\
    &\qquad\leq\frac{1}{m}(D(\tau_{X^mE^m}\|\widetilde{\rho}_{X^mE^m})+D(\tau_{E^m}\|\widetilde{\rho}_{E^m})
    +\left|mR-H(X^m|E^m)_\tau\right|^+) \,.
\end{align}
Minimizing over $\tau_{X^mE^m}$ gives $G(\widetilde{\rho}_{X^mE^m},mR)/m$, by the variational expression for $G$ established in Lemma~\ref{lem: variational G}.  Recalling the definition of $G$ in~\eqref{eq: definition G} proves the claim.
\end{proof}

Next, we show that blockwise pinching changes the error exponent
by at most a correction that vanishes as the block size tends to infinity.
\begin{lemma}
\label{lem:block-pinching}
Let  $\rho_{XE}$ be a classical-quantum state and let $\widetilde{\rho}_{X^mE^m}=\mathcal{P}_{\rho_E^{\otimes m}} \circ\mathcal{P}_{\omega_{E^m}}(\rho_{XE}^{\otimes m})$.
For every fixed $m\geq 1$, write $n=mk+\ell$, with $0\leq\ell<m$. Then,
\begin{align}
&\limsup_{n\to\infty}-\frac2n\log
\E_{H_n}
F\!\left(
\cR_{H_n}(\rho_{XE}^{\otimes n}),
\frac{I_{Z_n}}{|\mathcal Z_n|}\otimes\rho_E^{\otimes n}
\right)^\frac{1}{2}
\notag\\
&\quad\le
\limsup_{n\to\infty}-\frac2n\log
\E_{H_n}
F\!\left(
\cR_{H_n}(\widetilde{\rho}_{X^mE^m}^{\otimes k}\otimes \rho_{XE}^{\otimes \ell}),
\frac{I_{Z_n}}{|\mathcal Z_n|}\otimes\rho_E^{\otimes n}
\right)^\frac{1}{2}
+\frac{2}{m}(d_E-1)\log(m+1).
\label{eq:block-pinching-bound}
\end{align}
\end{lemma}

\begin{proof}
For $n=mk+\ell$, apply the pinching argument in Lemma~\ref{lem: pinching lemma} first to
$\cP_{\omega_{E^m}}^{\otimes k}$ and then to
$\cP_{\rho_E^{\otimes m}}^{\otimes k}$ on the first $mk$ copies, leaving the last $\ell$ copies untouched. We then use that
\begin{align}
    & \max\{|\mathrm{spec}(\rho_E^{\otimes m})|, |\mathrm{spec}(\omega_{E^{m}})|\} \leq (m+1)^{d_E-1} \,.
\end{align}
For every seed realization $h_n$,
\begin{align}
&-\log F\!\left(
\cR_{h_n}(\rho_{XE}^{\otimes n}),
\frac{I_{Z_n}}{|\mathcal Z_n|}\otimes\rho_E^{\otimes n}
\right)
\notag\\
&\quad\le
-\log F\!\left(
\cR_{h_n}(\widetilde{\rho}_{X^mE^m}^{\otimes k}\otimes \rho_{XE}^{\otimes \ell}),
\frac{I_{Z_n}}{|\mathcal Z_n|}\otimes\rho_E^{\otimes n}
\right)
+2k(d_E-1)\log(m+1).
\end{align}
Equivalently, the square root of the first fidelity is at least
$(m+1)^{-k(d_E-1)}$ times the square root of the second.  Average this pointwise inequality over $H_n$, apply $-2\log$, divide by $n$, and use $k/n\le1/m$.  Taking the limsup proves \eqref{eq:block-pinching-bound}.
\end{proof}

\begin{remark}
    The pinching channels are introduced only for the analysis and are not implemented by the protocol. Their dependence on the state therefore does not affect universality. The protocol consists of applying a hash function selected from a two-universal family, with both the family and the seed distribution chosen independently of $\rho_{XE}$.
\end{remark}

\begin{proposition}
\label{prop: universal bound sandwiched}
Let $\rho_{XE}$ be a classical-quantum state and $R\geq 0$. We have
    \begin{align}
         E_{\mathrm{pa}}^{\mathrm{s}}(\rho_{XE},R) \leq \max_{\alpha \in \left[\frac{1}{2},1\right]}\frac{1-\alpha}{\alpha}\left( R-\widetilde{H}_{\alpha}^{\lambda(\alpha)}(X|E)_{\rho} \right)\,,
    \end{align}
     where $\lambda(\alpha) = (1-2\alpha)/(1-\alpha)$. Moreover, every sequence of two-universal hash families satisfying
\(\lim_{n\to\infty}\frac{1}{n}\log|\mathcal{Z}_n|=R\)
achieves the bound. In particular, the hash family may be chosen
independently of the state $\rho_{XE}$.

\end{proposition}

\begin{proof}
Fix $m\geq 1$.  We first introduce the pinching using Lemma~\ref{lem:block-pinching}. We then use Lemma~\ref{lem: universal LE bound} which gives an upper bound in terms of the log-Euclidean conditional entropy of the pinched state and gets rid of the finite-size contribution $\rho_{XE}^{\otimes \ell}$. We also use that pinching does not change the marginal, that is,  $\tilde{\rho}_{E^m}=\rho_{E}^{\otimes m}$.  We obtain
\begin{align}
&\limsup_{n\to\infty}-\frac2n\log
\E_{H_n}
F\!\left(
\cR_{H_n}(\rho_{XE}^{\otimes n}),
\frac{I_{Z_n}}{|\mathcal Z_n|}\otimes\rho_E^{\otimes n}
\right)^{1/2}
\notag\\
&\quad\le
\max_{\alpha\in[1/2,1]}
\frac{1-\alpha}{\alpha}
\left(
R-\frac1mH_{\alpha,\infty}^{\lambda(\alpha)}(X^m|E^m)_{\widetilde\rho}
\right)
+\frac2m(d_E-1)\log(m+1).
\label{eq:before-entropy-comparison}
\end{align}
The right-hand side coincides with that in Lemma~\ref{lem: LE bound with pinched state}. Applying the entropy comparison from the proof of Proposition~\ref{prop: bound sandwiched} and letting \(m\to\infty\) yields the result.
\end{proof}

\section{Acknowledgments}
We thank Bartosz Regula for discussions. R.R. acknowledges financial support from the ERC grant GIFNEQ 101163938. M.T.\ acknowledge support from the National Research Foundation Investigatorship Award (NRF-NRFI10-2024-0006) and from the National Research Foundation, Singapore through the National Quantum Office, hosted in A*STAR, under its Centre for Quantum Technologies Funding Initiative (S24Q2d0009).

\bibliographystyle{ultimate}
\bibliography{my}

\newpage

\appendix

\section{Definitions for states that are not full-rank}
\label{app: non full-rank}
In Section~\ref{sec: notation}, we introduced the club-sandwiched conditional entropy. In what follows, we show that this definition is consistent with the extension obtained by taking limits of the corresponding objective function evaluated on full-rank states, which gives a consistent definition.

Let us denote
\begin{align}
   \widetilde{\Phi}_{\alpha,\lambda}(\rho_{AE},\sigma_E;\varepsilon) = \Tr[\left(\rho_{AE}(\varepsilon)^\frac{1}{2}\rho_E(\varepsilon)^{\frac{1-\lambda}{2}\frac{1-\alpha}{\alpha}}\sigma_E(\varepsilon)^{\lambda\frac{1-\alpha}{\alpha}}\rho_E(\varepsilon)^{\frac{1-\lambda}{2}\frac{1-\alpha}{\alpha}}\rho_{AE}(\varepsilon)^\frac{1}{2}\right)^\alpha] \,,
\end{align}
where we defined the depolarized states as
\begin{align}
    \rho_{AE}(\varepsilon)=(1-\varepsilon)\rho_{AE}+\varepsilon\pi_{AE}\,,\quad \rho_{E}(\varepsilon)=(1-\varepsilon)\rho_{E}+\varepsilon\pi_{E} \,,\quad \sigma_{E}(\varepsilon)=(1-\varepsilon)\sigma_{E}+\varepsilon\pi_{E}
\end{align}
Here, $\pi_{AE}=I_{AE}/d_{AE}$ is the fully mixed state. In the following, we show that the definition of the objective function is consistent with the limit $\varepsilon \rightarrow 0$ of full-rank states. In particular, whenever $\supp(\rho_{AE}) \subseteq \supp(I_A\otimes \sigma_E)$ the limit $\varepsilon \rightarrow 0$ is finite, and the objective function converges to the definition given in Section~\ref{sec: notation}, where negative powers are taken only on the support of the states. If the support condition is not met, the objective function is infinite.
\begin{lemma}
Let $\alpha \in (0,1)$, $\lambda <0$ and  $\rho_{AE}$ and $\sigma_E$ be quantum states. 
If $\supp(\rho_{AE}) \subseteq \supp(I_A\otimes \sigma_E)$ then
\begin{align}
     &\lim_{\varepsilon \rightarrow 0^+} \widetilde{\Phi}_{\alpha,\lambda}(\rho_{AE},\sigma_E;\varepsilon) = \Tr[\left(\rho_{AE}^\frac{1}{2}\rho_E^{\frac{1-\lambda}{2}\frac{1-\alpha}{\alpha}}\sigma_E^{\lambda\frac{1-\alpha}{\alpha}}\rho_E^{\frac{1-\lambda}{2}\frac{1-\alpha}{\alpha}}\rho_{AE}^\frac{1}{2}\right)^\alpha]
\end{align}
where negative powers are taken on the support of the states.
Moreover, if $\supp(\rho_{AE}) \not\subseteq \supp(I_A\otimes \sigma_E)$ then
    \begin{align}
        \lim_{\varepsilon \rightarrow 0^+} \widetilde{\Phi}_{\alpha,\lambda}(\rho_{AE},\sigma_E;\varepsilon) = +\infty \,.
    \end{align}
\end{lemma}
\begin{proof}
     Let us first consider the case $\supp(\rho_{AE}) \not\subseteq \supp(I_A\otimes \sigma_E)$. We denote with $R_{AE},R_{E},S_E$ the projectors onto the support of $\rho_{AE},\rho_E,\sigma_E$, respectively. Moreover, we denote with $R_{AE}^\perp,R_{E}^\perp,S_E^\perp$ the projectors onto the kernel. Since the proof for the above function follows with minor modifications, for brevity, we consider the limit $\varepsilon \rightarrow 0$ of the function
     \begin{align}
        &\widetilde{\Phi}'_{\alpha,\lambda}(\rho_{AE},\sigma_E;\varepsilon) \\
        &\; = \Tr[\left((\rho_{AE}+\varepsilon R_{AE}^\perp)^\frac{1}{2}(\rho_E+\varepsilon R_{E}^\perp)^{\frac{1-\lambda}{2}\frac{1-\alpha}{\alpha}}(\sigma_E+\varepsilon S_{E}^\perp)^{\lambda\frac{1-\alpha}{\alpha}}(\rho_E+\varepsilon R_{E}^\perp)^{\frac{1-\lambda}{2}\frac{1-\alpha}{\alpha}}(\rho_{AE}+\varepsilon R_{AE}^\perp)^\frac{1}{2}\right)^\alpha] \,.
     \end{align}
Next we use that $\rho_{AE}+\varepsilon R_{AE}^\perp\geq \rho_{AE}$, that $\Tr((AA^\dagger)^\alpha) = \Tr((AA^\dagger)^\alpha)$ and the fact that the trace functional $M \rightarrow \Tr[f(M)]$ inherits the monotonicity from $f$~(see e.g.,~\cite{carlen2010trace}). In addition, we compute the power on $\sigma_E + \varepsilon S_E^\perp$ using the fact that the two terms are orthogonal. We obtain the lower bound
\begin{align}
     \widetilde{\Phi}'_{\alpha,\lambda}(\rho_{AE},\sigma_E;\varepsilon)  &\geq \Tr[\left(\rho_{AE}^\frac{1}{2}(\rho_E+\varepsilon R_{E}^\perp)^{\frac{1-\lambda}{2}\frac{1-\alpha}{\alpha}}(\sigma_E^{\lambda\frac{1-\alpha}{\alpha}}+\varepsilon^{\lambda\frac{1-\alpha}{\alpha}} S_{E}^\perp)(\rho_E+\varepsilon R_{E}^\perp)^{\frac{1-\lambda}{2}\frac{1-\alpha}{\alpha}}\rho_{AE} ^\frac{1}{2}\right)^\alpha] \\
     & \quad \geq \varepsilon^{\lambda\frac{1-\alpha}{\alpha}}\Tr[\left(\rho_{AE}^\frac{1}{2}(\rho_E+\varepsilon R_{E}^\perp)^{\frac{1-\lambda}{2}\frac{1-\alpha}{\alpha}} S_{E}^\perp (\rho_E+\varepsilon R_{E}^\perp)^{\frac{1-\lambda}{2}\frac{1-\alpha}{\alpha}}\rho_{AE} ^\frac{1}{2}\right)^\alpha]\,,
\end{align}
where in the last line we used that $\sigma_E^{\lambda\frac{1-\alpha}{\alpha}}$ is a positive semidefinite operator. Next, we need to show that the term inside the trace is larger than a fixed constant for any $\varepsilon$. Since the operator inside the trace and under the power is positive semidefinite, we can equivalently show that it is non-zero. Hence, it is enough to prove that its trace is non-zero. We have
\begin{align}
    \Tr \big[\rho_{AE}^\frac{1}{2}(\rho_E+\varepsilon R_{E}^\perp)^{\frac{1-\lambda}{2}\frac{1-\alpha}{\alpha}} S_{E}^\perp (\rho_E+\varepsilon R_{E}^\perp)^{\frac{1-\lambda}{2}\frac{1-\alpha}{\alpha}}\rho_{AE} ^\frac{1}{2}\big] = \Tr\big[\rho_E^{{\frac{1-\lambda}{2}\frac{1-\alpha}{\alpha}}+1} S_E^\perp \big]\,.
\end{align}
It is then enough to show that the positive semidefinite operator $\rho_E S_E^\perp \rho_E$ is non-zero. The assumption $\supp(\rho_{AE}) \not\subseteq \supp(I_A\otimes \sigma_E)$ implies that there exists an eigenvector $\ket{\psi}$ of $\rho_{AE}$, with eigenvalue $\lambda>0$, such that $c := \bra{\psi} I_A \otimes S_E^\perp \ket{\psi} > 0$. This implies that the vector $I_A \otimes S_E^\perp \ket{\psi} \neq 0$. This shows that
$(I_A\otimes S_E^\perp) \rho_{AE}(I_{A}\otimes S_E^\perp) \neq 0$ since 
\begin{align}
    (I_A\otimes S_E^\perp) \rho_{AE}(I_{A}\otimes S_E^\perp) \geq \lambda (I_A\otimes S_E^\perp) \ketbra{\psi}{\psi}(I_{A}\otimes S_E^\perp) \neq 0 \,.
\end{align}
Hence, the partial trace of the positive semidefinite operator $(I_A\otimes S_E^\perp) \rho_{AE}(I_{A}\otimes S_E^\perp) $ must be non-zero. Its partial trace is equal to 
\begin{align}
    \Tr_A[ (I_A\otimes S_E^\perp) \rho_{AE}(I_{A}\otimes S_E^\perp)] =  S_E^\perp \rho_{E} S_E^\perp
\end{align}
Hence, $S_E^\perp \rho_{E} S_E^\perp \neq 0$ and $\Tr[S_E^\perp \rho_{E}] \neq 0$. Finally, since $\rho_{E}^{{\frac{1-\lambda}{2}\frac{1-\alpha}{\alpha}}+1}$ has the same support of $\rho_E$, it holds also that $c =\Tr \big[\rho_{E}^{\frac{1-\lambda}{2}\frac{1-\alpha}{\alpha}+1} S_E^\perp \big] \neq 0$ (and it is positive as it can be written as the trace of the positive semidefinite operator). Overall, we obtain that
\begin{align}
     \widetilde{\Phi}'_{\alpha,\lambda}(\rho_{AE},\sigma_E;\varepsilon)  \geq  \varepsilon^{\lambda\frac{1-\alpha}{\alpha}} c \,.
\end{align}
Since ${\lambda\frac{1-\alpha}{\alpha}} < 0$, we have that $ \lim_{\varepsilon \rightarrow 0}\widetilde{\Phi}'_{\alpha,\lambda}(\rho_{AE},\sigma_E;\varepsilon)\rightarrow +\infty$.

\bigskip

Let us turn to the case $\supp(\rho_{AE}) \subseteq \supp(I_A\otimes \sigma_E)$. Lemma~\ref{lem: support lemma}  implies that $\supp(I_A \otimes \rho_{E}) \subseteq \supp(I_A \otimes\sigma_E)$. We can expand the terms inside the trace and under the power of $\widetilde{\Phi}'_{\alpha,\lambda}(\rho_{AE},\sigma_E;\varepsilon) $ by taking advantage of the fact that the states are orthogonal to the projector onto their kernel. All terms that contain $\varepsilon$ raised to some positive power go to zero whenever $\varepsilon$ vanishes. We are then left to show that the remaining term also vanishes. By the chain of support inclusions $\supp(\rho_{AE}) \subseteq \supp(I_A\otimes \rho_E) \subseteq \supp(I_A\otimes \sigma_E)$ we have
\begin{align}
  &(\rho_{AE} + \varepsilon R_{AE}^\perp)^\frac{1}{2} (\rho_E+\varepsilon R_{E}^\perp)^{\frac{1-\lambda}{2}\frac{1-\alpha}{\alpha}}\varepsilon^{\lambda\frac{1-\alpha}{\alpha}} S_{E}^\perp(\rho_E+\varepsilon R_{E}^\perp)^{\frac{1-\lambda}{2}\frac{1-\alpha}{\alpha}} (\rho_{AE} + \varepsilon R_{AE}^\perp)^\frac{1}{2} \\
  & \quad \qquad = \varepsilon^{1+ 2 \frac{1-\lambda}{2} \frac{1-\alpha}{\alpha} + \lambda \frac{1-\alpha}{\alpha}} I_A \otimes  S_{E}^\perp \\
  & \quad \qquad = \varepsilon^{\frac{1}{\alpha}} I_A \otimes  S_{E}^\perp \,.
\end{align}
Since $\frac{1}{\alpha} >0$, this term also converges to zero. Therefore, no infinity appears, and hence the result.
\end{proof}

In Section~\ref{sec: notation}, we introduced the LE conditional entropy. In what follows, we show that this definition is consistent with the extension obtained by taking limits of the corresponding objective function evaluated on full-rank states.

Let us denote
\begin{align}
   \Phi_{\alpha,\lambda}(\rho_{AE},\sigma_E;\varepsilon) = \Tr [\exp\Big(\alpha\log{\rho_{AE}(\varepsilon)}+(1-\alpha)\big((1-\lambda) \log{\rho_E(\varepsilon)} + \lambda \log{\sigma_E(\varepsilon)}\big) \Big)] \,,
\end{align}
where we defined the depolarized states as
\begin{align}
    \rho_{AE}(\varepsilon)=(1-\varepsilon)\rho_{AE}+\varepsilon\pi_{AE}\,,\quad \rho_{E}(\varepsilon)=(1-\varepsilon)\rho_{E}+\varepsilon\pi_{E} \,,\quad \sigma_{E}(\varepsilon)=(1-\varepsilon)\sigma_{E}+\varepsilon\pi_{E}
\end{align}
Here, $\pi_{AE}=I_{AE}/d_{AE}$ is the fully mixed state. The following proof is inspired by the arguments in~\cite[Lemma III.1]{mosonyi2017strong} and~\cite[Lemma 4.1]{hiai1993golden}.

\begin{lemma}
Let $\alpha \in (0,1)$, $\lambda <0$ and  $\rho_{AE}$ and $\sigma_E$ be quantum states. 
If $\supp(\rho_{AE}) \subseteq \supp(I_A\otimes \sigma_E)$ then
\begin{align}
     &\lim_{\varepsilon \rightarrow 0^+} \Phi_{\alpha,\lambda}(\rho_{AE},\sigma_E;\varepsilon) =\Tr [R_{AE}\mathrm{exp}\Big(\alpha\log{(\rho_{AE})}+(1-\alpha)R_{AE}\big((1-\lambda) \log{(\rho_E}) + \lambda \log{(\sigma_E})\big)R_{AE} \Big)]
\end{align}
Here, $R_{AE}$ is the projector onto the support of $\rho_{AE}$, and the logarithms are taken only on the support of the states. Moreover, if $\supp(\rho_{AE}) \not\subseteq \supp(I_A\otimes \sigma_E)$ then
    \begin{align}
        \lim_{\varepsilon \rightarrow 0^+}  \Phi_{\alpha,\lambda}(\rho_{AE},\sigma_E;\varepsilon) = +\infty \,.
    \end{align}
\end{lemma}

\begin{proof}
We first consider the case $\supp(\rho_{AE}) \not\subseteq \supp(I_A\otimes \sigma_E)$. 
It is well known and straightforward to verify that $\supp(\rho_{AE}) \subseteq \supp(I_A\otimes \rho_E)$, which, for instance, guarantees the finiteness of the von Neumann conditional entropy. For brevity, since the proof for our specific setting can be treated analogously, we consider the limit
$\varepsilon \to 0$ of the function
\begin{align}
\Tr\!\left[
\exp\!\Big(
\alpha \log(\rho+\varepsilon I)
+\beta\log(\tau+\varepsilon I)
+\gamma \log(\sigma+\varepsilon I)\big)
\Big)
\right],
\end{align}
under the assumptions that $\alpha,\beta \geq 0$, $\gamma<0$,
$\operatorname{supp}(\rho)\subseteq \operatorname{supp}(\tau)$ and
$\operatorname{supp}(\rho)\nsubseteq \operatorname{supp}(\sigma)$.
 In the following, we denote by $R$, $T$, and $S$ the orthogonal projectors onto the supports of $\rho$, $\tau$, and $\sigma$, respectively. Correspondingly, we denote by $R^\perp$, $T^\perp$, and $S^\perp$ the orthogonal projectors onto their kernels.
By the assumption $\operatorname{supp}(\rho)\nsubseteq \operatorname{supp}(\sigma)$, there exists an eigenvector $\ket{\psi}$ of $\rho$, with eigenvalue $\lambda>0$, such that $c := \bra{\psi}S^\perp \ket{\psi} > 0$. 
We first note that since $\tau +\varepsilon I \geq (\lambda_{\min}(\tau)+\varepsilon) T + \varepsilon T^\perp$, the operator monotonicity of the logarithm implies that
$\log(\tau+\varepsilon I) \geq (\log{\big(\lambda_{\min}(\tau)+\varepsilon\big)) T +\log{\varepsilon} T^\perp}$. Since $\operatorname{supp}(\rho)\subseteq \operatorname{supp}(\tau)$, we have that
\begin{align}
    \bra{\psi}\log(\tau+\varepsilon I) \ket{\psi} \geq \log{\big(\lambda_{\min}(\tau)+\varepsilon\big)} \,.
\end{align}
We then have
\begin{align}
    &\Tr\!\left[
\exp\!\big(
\alpha \log(\rho+\varepsilon I)
+\beta\log(\tau+\varepsilon I)
+\gamma \log(\sigma+\varepsilon I)
\big)
\right]\\
& \quad\geq \bra{\psi} \exp\!\big(
\alpha \log(\rho+\varepsilon I)
+\beta\log(\tau+\varepsilon I)
+\gamma \log(\sigma+\varepsilon I)
\big) \ket{\psi} \\
& \quad\geq \exp\!\left(\bra{\psi} 
\left(\alpha \log(\rho+\varepsilon I)
+\beta\log(\tau+\varepsilon I)
+\gamma\log(\sigma+\varepsilon I)\right)
 \ket{\psi} \right) \\
 & \quad =\exp\!\left(\alpha \log(\lambda+\varepsilon )
+\beta\bra{\psi}\log(\tau+\varepsilon I) \ket{\psi}
+\gamma \bra{\psi} \log(\sigma+ \varepsilon S) \ket{\psi} +\gamma \bra{\psi} \log(\varepsilon) S^\perp \ket{\psi}
  \right) \\
 & \quad \geq  (\lambda+\varepsilon)^\alpha (\lambda_{\min}(\tau)+\varepsilon)^{\beta} \exp(\gamma \bra{\psi} \log(\sigma+ \varepsilon S) \ket{\psi}) \varepsilon^{c\gamma} \,
\end{align}
where the second inequality follows from the convexity of the exponential function.
Since $\varepsilon^{c\gamma} \rightarrow +\infty$ and all other terms remain finite, we obtain the desired result.

\bigskip

Let us turn to the case $\supp(\rho_{AE}) \subseteq \supp(I_A\otimes \sigma_E)$. Lemma~\ref{lem: support lemma}  implies that $\supp(I_A \otimes \rho_{E}) \subseteq \supp(I_A \otimes\sigma_E)$. For brevity, since the proof for our specific setting can be treated with minor modifications, we consider the limit $\varepsilon \rightarrow 0$ of the function
\begin{align}
\Tr\!\left[
\exp\!\Big(
\alpha \log \big(\rho+\varepsilon R^\perp\big)
+\beta\log\big(\tau+\varepsilon T^\perp\big)
+\gamma \log\big(\sigma+\varepsilon S^\perp\big)
\Big)
\right],
\end{align}
under the assumption that $\alpha>0$,$\beta \geq 0$, $\gamma<0$, $\alpha+\beta+\gamma=1$, $\supp(\rho)\subseteq \supp(\tau)\subseteq \supp(\sigma)$.
For $\varepsilon>0$ let
 \begin{align}
     C(\varepsilon)
= \alpha \log \big(\rho+\varepsilon R^\perp\big)
+\beta\log\big(\tau+\varepsilon T^\perp\big)
+\gamma \log\big(\sigma+\varepsilon S^\perp\big)
 \end{align}
Then,
\begin{align}
    C(\varepsilon)
    = \alpha (\log \rho)\,R + \beta (\log \sigma)\,S + \gamma (\log \tau)\,T
    + (\log \varepsilon)\,(R^\perp + S^\perp + T^\perp)\, .
\end{align}
We decompose the Hilbert space according to the orthogonal decomposition
\(
\mathcal H = \mathrm{supp}(\rho) \oplus \ker(\rho)
\),
with corresponding projectors \(R\) and \(R^\perp\).
Inserting the resolution of the identity
\(I = R + R^\perp\) on both sides of \(C(\varepsilon)\) and expanding the resulting terms,
we obtain the following block-matrix representation of \(C(\varepsilon)\) with respect to this decomposition
\begin{align}
    C(\varepsilon)
    =
    \begin{pmatrix}
        C_{11} & C_{12} \\
        C_{12}^\dagger & C_{22} + D_{22}(\varepsilon)
    \end{pmatrix} \, .
\end{align}
where
\begin{align*}
&C_{11} = R(\alpha\log \rho + \beta\log \tau+\gamma \log{\sigma})R, \\
&C_{12} = R(\beta \log{\tau}+\gamma (\log{\sigma})S)R^\perp, \\
&C_{22} = (T-R)\beta\log{\tau}(T-R)+(S-R)\gamma\log{\sigma}(S-R) \\
&D_{22}(\varepsilon) = (\log \varepsilon)(\alpha R^\perp + \beta T^\perp+\gamma S^\perp) .
\end{align*}
Since $\alpha R^\perp + \beta T^\perp+\gamma S^\perp$ is positive on $\operatorname{ran}(P)^\perp$, there exists
$\delta = \min\{\alpha,\alpha+\beta,\alpha+\beta+\gamma\} =\alpha>0$ such that
\begin{align}
    \alpha R^\perp + \beta T^\perp+\gamma S^\perp \geq \delta R^\perp
\end{align}
Let \(k := \mathrm{rank}(R)\), and assume \(0 < k < n\); the boundary cases \(k = 0\) and
\(k = n\) are trivial and can be treated separately.  
We order the eigenvalues of \(C(\varepsilon)\), counted with their algebraic multiplicities,
in nonincreasing order,
\begin{align}
    \lambda_1(\varepsilon) \ge \cdots \ge \lambda_k(\varepsilon)
\ge \lambda_{k+1}(\varepsilon) \ge \cdots \ge \lambda_n(\varepsilon),
\end{align}
and denote by
\begin{align}
    u_1(\varepsilon), \ldots, u_k(\varepsilon),
u_{k+1}(\varepsilon), \ldots, u_n(\varepsilon)
\end{align}
the corresponding orthonormal family of eigenvectors.
Then
\begin{equation}
\label{eq: eig decomposition}
e^{C(\varepsilon)} = \sum_{i=1}^n e^{\lambda_i(\varepsilon)}\,u_i(\varepsilon)u_i(\varepsilon)^\ast .
\end{equation}
Let \(\mu_1 \ge \cdots \ge \mu_k\) denote the eigenvalues of the upper-left block \(C_{11}\),
ordered nonincreasingly.  
Standard majorization results (see, e.g.,~\cite[Corollary~7.2]{bhatia1997matrix}) implies the inequality
\begin{equation}
\label{eq: majorization result}
    \sum_{i=1}^l \mu_i \le \sum_{i=1}^l \lambda_i(\varepsilon),
    \qquad 1 \le l \le k .
\end{equation}
Moreover, since
\begin{align}
\lim_{\varepsilon \downarrow 0}
\frac{C(\varepsilon)}{\log \varepsilon}
=
\begin{pmatrix}
    0 & 0 \\
    0 & \alpha R^\perp+\beta T^\perp+\gamma S^\perp
\end{pmatrix},
\end{align}
it follows that, for indices \(k < i \le n\), the ratio
\(\lambda_i(\varepsilon)/\log \varepsilon\) converges to a strictly positive limit as
\(\varepsilon \downarrow 0\). Consequently, $\lambda_i(\varepsilon) \longrightarrow -\infty$ for $ k < i \le n$.
Invoking~\eqref{eq: eig decomposition}, it is enough to verify that for every
sequence \((\varepsilon_n)_{n\in\mathbb N}\) with \(\varepsilon_n \rightarrow 0^+\), there exists a converging subsequence \((\varepsilon_n')_{n\in\mathbb N}\), such that, for each \(1 \le i \le k\), it converges to
\begin{align}
\label{eq: supp verify 1}
    &\lambda_i(\varepsilon_n') \to \mu_i, \\
\label{eq: supp verify 2}
    &u_i(\varepsilon_n') \to u_i \in \operatorname{ran}(R), \\
\label{eq: supp verify 3}
    &C u_i = \mu_i u_i .
\end{align}
Hence, passing to a subsequence if necessary, we may assume that
\(u_i(\varepsilon_n) \to u_i\) for all \(1 \le i \le k\), and that
\(\{u_1,\dots,u_k\}\) forms an orthonormal family.
Define the components
\begin{align}
    u_i^{(1)} := R\,u_i(\varepsilon_n),
\qquad
u_i^{(2)} := R^\perp u_i(\varepsilon_n).
\end{align}
Then
\begin{align}
\lambda_i(\varepsilon_n)
&= \big\langle C(\varepsilon_n)u_i(\varepsilon_n),\,u_i(\varepsilon_n)\big\rangle \\
&= \langle C_{11}u_i^{(1)},u_i^{(1)}\rangle
   + 2\,\Re\langle C_{12}u_i^{(2)},u_i^{(1)}\rangle
   + \langle (C_{22}+D_{22}(\varepsilon_n))u_i^{(2)},u_i^{(2)}\rangle \\
   \label{eq: bound eigenvalues eps}
&\le \langle C_{11}u_i^{(1)},u_i^{(1)}\rangle
   + 2\,\Re\langle C_{12}u_i^{(2)},u_i^{(1)}\rangle
   + \langle C_{22}u_i^{(2)},u_i^{(2)}\rangle
   + (\log \varepsilon_n)\,\delta\,\|u_i^{(2)}\|^2 ,
\end{align}
where we used the bound
\(D_{22}(\varepsilon_n) \le (\log \varepsilon_n)\,\delta\,R^\perp\).
Since \(\mu_1 \le \lambda_1(\varepsilon_n)\) by~\eqref{eq: majorization result} and hence the eigenvalues must be bounded from below, relation~\eqref{eq: bound eigenvalues eps} implies that it must be that
\(u_1^{(2)} \to 0\), and hence \(u_1^{(1)} \to u_1 \in \operatorname{ran}(R)\).
Moreover, the same relation implies that
\begin{align}
    \mu_1 \leq \limsup_{n\to\infty} \lambda_1(\varepsilon_n)
\le \langle C_{11}u_1,u_1\rangle \le \mu_1,
\end{align}
which shows that
\(\lambda_1(\varepsilon_n) \to \langle C_{11}u_1,u_1\rangle = \mu_1\) and hence that converging subsequences must converge to a fixed value $\mu_1$.
In addition, since \(\mu_1\) is the maximal eigenvalue of \(C_{11}\), it follows that
\(C_{11}u_1 = \mu_1 u_1\).
If \(k \ge 2\), the same reasoning applies to \(\lambda_2(\varepsilon_n)\), which is
uniformly bounded from below by~\eqref{eq: majorization result}.
Thus \(u_2^{(2)} \to 0\) and \(u_2^{(1)} \to u_2 \in \operatorname{ran}(R)\).
Using again~\eqref{eq: majorization result} and~\eqref{eq: bound eigenvalues eps}, we conclude that
\begin{align}
    \lambda_2(\varepsilon_n) \to \langle C_{11}u_2,u_2\rangle = \mu_2 \,.
\end{align}
Since, as shown above, \(u_2\) is orthogonal to \(u_1\), it must belong to the subspace
\(\{u_1\}^\perp \cap \operatorname{ran}(R)\); hence \(C_{11}u_2=\mu_2u_2\), where \(\mu_2\)
is the largest eigenvalue of \(C_{11}\) restricted to this subspace. Repeating the same
argument inductively yields\eqref{eq: supp verify 1},~\eqref{eq: supp verify 2} and~\eqref{eq: supp verify 3}
for all \(1\le k\le n\).

\end{proof}

\section{Properties of the club-sandwiched conditional entropy}
\label{app: properties of lambda-conditional entropies}
The first lemma establishes a duality relation for the club-sandwiched conditional entropy. 
Specifically, it shows that this quantity is dual to the $\alpha$--$z$ conditional entropy arrow up, defined as
\begin{align}
H_{\alpha,z}^{\uparrow}(A|E)_{\rho}
:= \sup_{\sigma_{E}} - D_{\alpha,z}(\rho_{AE}\,\|\, I_A \otimes \sigma_E),
\end{align}
where the $\alpha$--$z$ R\'enyi relative entropy is given by~\cite{audenaert13_alphaz}
\begin{align}
D_{\alpha,z}(\rho \| \sigma)
:= \frac{1}{\alpha-1}
\log \Tr\!\left[\left(\rho^{\frac{\alpha}{2z}}
\sigma^{\frac{1-\alpha}{z}}
\rho^{\frac{\alpha}{2z}}\right)^{z}\right].
\end{align}
This is a specific case of the result established in~\cite[Theorem 3.1]{rubboli2024quantum}.
\begin{lemma}
 \label{lem: duality relation}
Let $\rho_{AEC}$ be a pure quantum state. Then, when $\alpha \in [1/2,1]$, we have
    \begin{align}
\widetilde{H}_\alpha^{\lambda(\alpha)}(A|E)_{\rho} =  -H_{\beta(\alpha),\beta(\alpha)/2}^{\uparrow}(A|C)_{\rho} \,,
    \end{align}
where $\lambda(\alpha)=(1-2\alpha)/(1-\alpha)$ and $\beta(\alpha)=2\alpha/(3\alpha-1)$.
\end{lemma}

\begin{lemma}
\label{lem: limit of lambda conditional entropy}
    Let $\rho_{AE}$ be a quantum state. Then, we have that
    \begin{align}
        \lim_{\alpha \rightarrow \frac{1}{2}^+} \widetilde{H}_\alpha^{\lambda(\alpha)}(A|E)_{\rho} = \widetilde{H}^\downarrow_\frac{1}{2}(A|E)_{\rho} \,, \quad \text{and} \qquad \lim_{\alpha \rightarrow 1^-} \widetilde{H}_\alpha^{\lambda(\alpha)}(A|E)_{\rho} = H(A|E)_{\rho} \,,
    \end{align}
    where $\lambda(\alpha)=(1-2\alpha)/(1-\alpha)$.
\end{lemma}
\begin{proof}
In both cases, we prove both an upper and a lower bound. 

Let us start with the case $\alpha \rightarrow 1/2$. The monotonicity in $\lambda$ established in~\cite[Proposition 8.3]{rubboli2024quantum} and the fact that $\frac{1-2\alpha}{1-\alpha} \leq  0$ imply that
\begin{align}
    \lim_{\alpha\rightarrow \frac{1}{2}^+} \widetilde{H}_\alpha^{\lambda(\alpha)}(A|E)_{\rho} \leq \lim_{\alpha\rightarrow \frac{1}{2}^+} \widetilde{H}_\alpha^{\downarrow}(A|E)_{\rho} = \widetilde{H}_{\frac{1}{2}}^{\downarrow}(A|E)_{\rho} \,.
\end{align}
Next, we prove a lower bound. In particular, we leverage the duality relation in Lemma~\ref{lem: duality relation}. 
Consider a purification $\rho_{AEC}$ of $\rho_{AE}$. We have 
\begin{align}
    \lim_{\alpha\rightarrow \frac{1}{2}^+} \widetilde{H}_\alpha^{\lambda(\alpha)}(A|E)_{\rho} = \lim_{\alpha\rightarrow \frac{1}{2}^+ }-H_{\beta(\alpha),\beta(\alpha)/2}^{\uparrow}(A|C)_{\rho} \geq \lim_{\alpha\rightarrow \frac{1}{2}^+} -H_{\beta(\alpha),1}^{\uparrow}(A|C)_{\rho} = -H_{2,1}^{\uparrow}(A|C)_{\rho} = \widetilde{H}_{\frac{1}{2}}^{\downarrow}(A|E)_{\rho} \,.
\end{align}
Here, the inequality follows from the monotonicity in $z$ for $\alpha>1$ (see~\cite[Proposition 10.2]{rubboli2024quantum}). Moreover, we use that $\beta(\alpha)\rightarrow 2^-$ for $\alpha \rightarrow1/2^+$ and that $\beta(\alpha)/2 \leq 1$.

Let us now turn to the case $\alpha \rightarrow 1^-$. Let us start with the upper bound. The monotonicity in $\lambda$ established in~\cite[Proposition 8.3]{rubboli2024quantum} and the fact that $(1-2\alpha)/(1-\alpha) \leq 0$ for $\alpha \geq 1/2$ imply that
   \begin{equation}
       \limsup_{\alpha \rightarrow 1^-} \widetilde{H}^{\lambda(\alpha)}_{\alpha}(A|E)_{\rho} \leq \limsup_{\alpha \rightarrow 1^-} \widetilde{H}^\downarrow_{\alpha}(A|E)_{\rho} = H(A|E)_{\rho} \,,
   \end{equation}
   where the latter inequality follows from the fact that the sandwiched divergence converges to the Umegaki relative entropy for $\alpha\rightarrow 1$.

Next, consider a purification $\rho_{AEC}$ of $\rho_{AE}$. We have that
    \begin{align}
        \liminf_{\alpha \rightarrow 1^-} \widetilde{H}_\alpha^{\lambda(\alpha)}(A|E)_{\rho} = \liminf_{\alpha \rightarrow 1^-} -H_{\beta(\alpha),\beta(\alpha)/2}^{\uparrow}(A|C)_{\rho} \geq \liminf_{\alpha \rightarrow 1^-} -H_{\beta(\alpha),1}^{\uparrow}(A|C)_{\rho} =-H(A|C)_{\rho} = H(A|E)_{\rho}
    \end{align}
    Here, the inequality follows from the fact that the $\alpha$-$z$ Rényi relative entropy is monotonically decreasing in $z$ for $\alpha>1$~\cite{lin2015investigating} and $\beta(\alpha)/2 \leq 1$. The final equality follows from the known fact that the Petz conditional entropy arrow up converges to the von Neumann conditional entropy (see e.g.~\cite[Appendix A]{rubboli2022new}).
\end{proof}

\begin{lemma}
\label{lem: lower bound lambda}
    Let $\rho_{XE}$ be a quantum state and $\lambda(\alpha)=(1-2\alpha)/(1-\alpha)$. Then for all $\alpha \in [1/2,1]$ it holds that $\widetilde{H}_{\alpha}^{\lambda(\alpha)}(X|E)_\rho \geq H(X|E)_\rho$.
\end{lemma}
\begin{proof}
    We use the duality relation of Lemma~\ref{lem: duality relation}.
Consider a purification $\rho_{AEC}$ of $\rho_{AE}$. We have
\begin{align}
     \widetilde{H}_\alpha^{\lambda(\alpha)}(A|E)_{\rho} = -H_{\beta(\alpha),\beta(\alpha)/2}^{\uparrow}(A|C)_{\rho} \geq  -H_{\beta(\alpha),1}^{\uparrow}(A|C)_{\rho} \geq -H(A|C)_{\rho} = H(A|E)_\rho \,.
\end{align}
Here, the first inequality follows from the monotonicity in $z$ for $\alpha>1$ (see~\cite[Proposition 10.2]{rubboli2024quantum}). The second inequality follows from the fact that $\beta(\alpha) \geq 1$ and the antimonotonicity in $\alpha$ (see~\cite[Proposition 10.1]{rubboli2024quantum}). The last equality is the duality relation of the von Neumann conditional entropy.
\end{proof}

The next asymptotic characterization of the club-sandwiched conditional entropy has been proved in~\cite[Proposition 8.1]{rubboli2024quantum}.

\begin{lemma}
\label{lem: old asymptotic}
Let $\rho_{AE}$ be a bipartite quantum state.
    For $1-\alpha/(1-\alpha) \leq \lambda \leq 0$ and $\alpha \in (0,1)$, it holds that
\begin{align}
        \lim_{n \rightarrow \infty} \frac{1}{n}\frac{1}{1-\alpha}\log{\Tr(\omega_{E^n}^{\frac{\lambda}{2}\frac{1-\alpha}{\alpha}}\rho_{E^n}^{\frac{1-\lambda}{2}\frac{1-\alpha}{\alpha}}\rho_{A^n E^n }\rho_{E^n}^{\frac{1-\lambda}{2}\frac{1-\alpha}{\alpha}}\omega_{E^n}^{\frac{\lambda}{2}\frac{1-\alpha}{\alpha}})^\alpha} = \widetilde{H}^\lambda_\alpha(A|E)_{\rho_{AE}}
    \end{align}
    where $\rho_{E^n} = \rho_E^{\otimes n}$ and $\rho_{A^n E^n} = \rho_{AE}^{\otimes n}$ and $\omega_{E^n}$ is the universal state of Lemma~\ref{lem: universal state}.
\end{lemma}

The next lemma concerns the symmetry of the optimizer of the LE conditional entropy.
\begin{lemma}
\label{lem: permutation invariant optimum}
    Let $\alpha \in (0,1)$ and $\lambda \leq 0$ and $\rho_{XE}$ be a quantum state. The optimizer in 
    \begin{align}
&H_{\alpha,\infty}^\lambda(X^m|E^m)_{\mathcal{P}_{\rho_E^{\otimes m}}\circ \mathcal{P}_{\omega_{E^m}}(\rho_{XE}^{\otimes m})}\notag\\
&\; =\inf_{\sigma_{E^m}}\frac{1}{1-\alpha}\log\Tr\Big[R_{X^m E^m}\exp\Big(\alpha\log{\big(\mathcal{P}_{\rho_E^{\otimes m}}\circ \mathcal{P}_{\omega_{E^m}}(\rho_{XE}^{\otimes m})\big)}\notag\\
& \hspace{10em}+(1-\alpha)R_{X^m E^m}\big((1-\lambda)\log{\rho_E^{\otimes m}}+\lambda\log{\sigma_{E^m}}\big)R_{X^m E^m}\Big)\Big] 
    \end{align}
can be chosen to be permutation invariant under the exchange of the $E$ systems.
    \end{lemma}
\begin{proof}
    The above objective function is convex in $\sigma_{E^m}$. To see this, we can use the variational formula in the proof of Lemma~\ref{lem: variational LE}
    \begin{align}
        & -\log\Tr\Big[R_{X^m E^m}\exp\Big(\alpha\log{\big(\mathcal{P}_{\rho_E^{\otimes m}}\circ \mathcal{P}_{\omega_{E^m}}(\rho_{XE}^{\otimes m})\big)}\notag\\
& \hspace{10em}+(1-\alpha)R_{X^m E^m}\big((1-\lambda)\log{\rho_E^{\otimes m}}+\lambda\log{\sigma_{E^m}}\big)R_{X^m E^m}\Big)\Big] \\& =\inf_{\tau_{X^m E^m}}\Big\{\Tr[\tau_{X^m E^m}\log{\tau_{X^m E^m}}] \\
    & \qquad    - \Tr[\tau_{X^m E^m}(\alpha\log{\big(\mathcal{P}_{\rho_E^{\otimes m}}\circ \mathcal{P}_{\omega_{E^m}}(\rho_{XE}^{\otimes m})\big)}+(1-\alpha)R_{X^m E^m}\big((1-\lambda)\log{\rho_E^{\otimes m}}+\lambda\log{\sigma_{E^m}}\big)R_{X^m E^m})]\Big\}
    \end{align}
Since the logarithm is operator concave and the infimum of concave functions is concave, the latter function is concave in $\sigma_{E^m}$. Due to the minus sign, the original function is convex.

Next, we leverage convexity to show that for a given $\sigma_{E^m}$, the symmetrized state 
\begin{align}
\sigma_{E^m}^{\mathrm{sym}} = \frac{1}{m!}\sum_{\pi\in S_m}
U_\pi^{E^m}\sigma_{E^m}(U_\pi^{E^m})^\dagger
\end{align}
performs at least as well as any state $\sigma_{E^m}$. 
We have 
\begin{align}
    &\log\Tr[R_{X^m E^m}\exp\Big(\alpha\log{\big(\mathcal{P}_{\rho_E^{\otimes m}}\circ \mathcal{P}_{\omega_{E^m}}(\rho_{XE}^{\otimes m})\big)}+(1-\alpha)R_{X^m E^m}\left((1-\lambda)\log{\rho_E^{\otimes m}}+\lambda\log{\sigma_{E^m}^{\mathrm{sym}}}\right)R_{X^m E^m}\Big)] \\
&  \leq \frac{1}{m!}\sum_{\pi\in S_m} \log\Tr \Big[R_{X^m E^m}\exp\Big(\alpha\log{\big(\mathcal{P}_{\rho_E^{\otimes m}}\circ \mathcal{P}_{\omega_{E^m}}(\rho_{XE}^{\otimes m})\big)}\\
&\qquad \qquad \qquad +(1-\alpha)R_{X^m E^m}\Big((1-\lambda)\log \rho_E^{\otimes m}+\lambda\log  U_\pi^{E^m} \sigma_{E^m} (U_\pi^{E^m})^\dagger \Big) R_{X^m E^m}\Big)\Big] \\
&  = \frac{1}{m!}\sum_{\pi\in S_m} \log\Tr \Big[R_{X^m E^m}\exp\Big(\alpha\log{\big(\mathcal{P}_{\rho_E^{\otimes m}}\circ \mathcal{P}_{\omega_{E^m}}(\rho_{XE}^{\otimes m})\big)}\\
&\, +(1-\alpha)R_{X^m E^m}\Big((1-\lambda)\log (I_{X^m}\otimes \rho_E^{\otimes m})+\lambda\log  \big((U_\pi^{X^m} \otimes U_\pi^{E^m}) (I_{X^m} \otimes \sigma_{E^m}) (U_\pi^{X^m} \otimes U_\pi^{E^m})^\dagger \big)\Big) R_{X^m E^m}\Big)\Big] \\
&  = \frac{1}{m!}\sum_{\pi\in S_m} \log\Tr \Big[R_{X^m E^m}\exp\Big( \big(U_\pi^{X^m} \otimes U_\pi^{E^m}\big)\Big(\alpha\log{\big(\mathcal{P}_{\rho_E^{\otimes m}}\circ \mathcal{P}_{\omega_{E^m}}(\rho_{XE}^{\otimes m})\big)}\\
&\qquad \qquad \qquad +(1-\alpha)R_{X^m E^m}\big((1-\lambda)\log \rho_E^{\otimes m}+\lambda\log  \sigma_{E^m}\big)  R_{X^m E^m} \Big) \big(U_\pi^{X^m}\otimes U_\pi^{E^m}\big)^\dagger\Big)\Big] \\
& = \log\Tr \Big[R_{X^m E^m}\exp\Big( \alpha\log{\big(\mathcal{P}_{\rho_E^{\otimes m}}\circ \mathcal{P}_{\omega_{E^m}}(\rho_{XE}^{\otimes m})\big)}\\
&\qquad \qquad \qquad +(1-\alpha)R_{X^m E^m}\Big((1-\lambda)\log \rho_E^{\otimes m}+\lambda\log  \sigma_{E^m}\Big)  R_{X^m E^m} \Big)\Big]
\end{align}
Here, the first inequality follows from the convexity in $\sigma_{E^m}$ established earlier. 
In the second equality, we insert the unitary implementing a permutation of the systems $X^m$ 
and use that $U_\pi^{X^m}(U_\pi^{X^m})^\dagger = I_{X^m}$. 
In the third equality, we exploit the fact that all remaining operators inside the exponential 
are invariant under conjugation by the product unitary $U_\pi^{X^m}\otimes U_\pi^{E^m}$, 
since they are invariant under permutations on the $m$ copies of the composite system $XE$. 
Finally, in the last equality, we pull the unitaries out of the exponential and use the 
unitary invariance of the trace.
\end{proof}

\section{Useful facts}
\label{app: lemmas}

\begin{lemma}
\label{lem: support lemma}
Let $\rho_{AE}$ and $\sigma_E$ be two quantum states. Then, it holds that $\supp(\rho_{AE})\subseteq\supp(I_A \otimes \rho_E)$. Moreover, if $\supp(\rho_{AE})\subseteq \supp(I_A \otimes \sigma_E)$, then it also holds that $\supp(I_A \otimes \rho_E)\subseteq \supp(I_A\otimes \sigma_E)$.
\end{lemma}
\begin{proof}
It is well known and straightforward to
verify that $\supp(\rho_{AE})\subseteq \supp(I_A \otimes \rho_E)$, which, for instance, guarantees the finiteness of the von Neumann
conditional entropy.

 Let us assume that $\supp(\rho_{AE}) \subseteq \supp(I_A\otimes \sigma_E)$. We need to prove that the latter condition implies $\supp(\rho_{E}) \subseteq \supp(\sigma_E)$. Let us denote with $S_E$ the projector onto the support of $\sigma_E$. The condition $\supp(\rho_{AE}) \subseteq \supp(I_A\otimes \sigma_E)$ implies that $(I_A\otimes S_E)\rho_{AE}(I_A\otimes S_E)$. Then, we have that
\begin{align}
    \rho_E=\Tr_A[\rho_{AE}] = \Tr_A[(I_A\otimes S_E)\rho_{AE}(I_A \otimes S_E)] = S_E \rho_E S_E \,,
\end{align}
which implies that $\supp(\rho_{E}) \subseteq \supp(\sigma_E)$.
\end{proof}

\begin{lemma}
\label{lem: pinching with identity}
Let $\rho_E$ be a quantum state. Then, we have that
    \begin{align}
         \mathcal{P}_{{I_A}\otimes \rho_E}(\cdot) = I_A \otimes \mathcal{P}_{\rho_E}(\cdot) \,.
    \end{align}
\end{lemma}
\begin{proof}
We have $I_A \otimes \rho_E = I_A \otimes (\sum_x \Pi_{E,x} \lambda_x) = \sum_x \lambda_x I_A \otimes \Pi_{E,x}$. Hence, we have that
    \begin{align}
        \mathcal{P}_{I_A \otimes \rho_E}(\cdot) = \sum_x I_A \otimes \Pi_{E,x} (\cdot) I_A \otimes \Pi_{E,x} = I_A \otimes \mathcal{P}_{\rho_E}(\cdot) \,.
    \end{align}
    \end{proof}
Note that the latter pinching is a map acting only on the side information $E$. 
Moreover, pinching with $\rho_E$ does not change the $E$ marginal of $\rho_{AE}$.
\begin{lemma}
\label{lem: marginal of commuting pinching}
   Let $\rho_{AE}$ be a bipartite quantum state and let $\omega_E$ a state such that $[\omega_E,\rho_E] = 0$.  Then, we have that 
    \begin{align}
        \Tr_{A}[\mathcal{P}_{I_A \otimes \omega_E}(\rho_{AE})] 
         =\rho_E
    \end{align}
\end{lemma}
\begin{proof}
    Using Lemma~\ref{lem: pinching with identity} we obtain that
    \begin{align}
        \Tr_{A}[\mathcal{P}_{I_A \otimes \omega_E}(\rho_{AE})] = \Tr_{A}[I_A \otimes \mathcal{P}_{\omega_E}(\rho_{AE})]
        = \mathcal{P}_{\omega_E}(\rho_E) =\rho_E
    \end{align}
\end{proof}

The following lemma will be instrumental in deriving a variational characterization of the LE conditional entropy. It is commonly referred to in the literature as the Gibbs variational principle and was originally established in~\cite{petz1988variational} (see also~\cite[Lemma~1.2]{hiai1993golden} for a simpler proof in finite dimensions). In our setting, however, we require a slight modification of this result, as the optimization is restricted to states supported on a fixed reference state. This constraint leads to the appearance of an additional support projector, as already observed in~\cite[Theorem~III.6]{mosonyi2017strong}.
\begin{lemma}[Constrained Gibbs variational principle]
\label{lem: Gibbs principle}
Let $H$ be a Hermitian operator and let $P$ a non-zero orthogonal projector. Then, 
\begin{align}
    -\log{\Tr[P\exp(-PHP)]}= \min_{\supp(\tau) \subseteq \supp(P)} \{\Tr[\tau\log{\tau}]+\Tr[\tau H]\}\,,
\end{align}
where the infimum is performed over all states $\tau$ whose support is contained in the one of $P$. Moreover, the minimum is attained for
\begin{align}
    \tau = \frac{P\exp(-PHP)}{\Tr[P \exp(-PHP)]}\,.
\end{align}
    
\end{lemma}

The following result relates to complex interpolation theory and offers a version of Hadamard's three-line theorem specifically for Schatten norms, as established in~\cite{beigi13_sandwiched}.
\begin{lemma}
\label{Hadamard}
Let $F:S \rightarrow L(\mathcal{H})$ be a bounded map that is holomorphic in the interior of $S$ and continuous up to the boundary. Assume that $1 \leq p_0,p_1 \leq \infty$ and for $0<\theta<1$ define $p_\theta$ by
\begin{equation}
\frac{1}{p_\theta} = \frac{1-\theta}{p_0}+\frac{\theta}{p_1}.
\end{equation}
For $k=0,1$, define 
\begin{equation}
M_k=\sup_{t \in \mathbb{R}}\|F(k+it)\|_{p_k}.
\end{equation}
Then, we have
\begin{equation}
\|F(\theta)\|_{p_{\theta}}\leq M_0^{1-\theta}M_1^\theta.
\end{equation}
\end{lemma}

The next lemma collects the key properties of the universal state~\cite[Lemma 1]{hayashitomamichel15c}. The conditional entropies can be defined in terms of the universal state, and its properties will be essential for our proofs.
\begin{lemma}
\label{lem: universal state}
Let $A$ be a system with $|A|=d$. For any $n \in \mathbb{N}$, there exists a state $\omega_{A^n} \in \mathcal{S}(A^n)$, referred to as the universal state, such that:
\begin{itemize}
\item For any permutation-invariant state $\tau_{A^n} \in \mathcal{S}(A^n)$, we have $\tau_{A^n} \leq g_{n,d} \omega_{A^n}$ with $g_{n,d} \leq (n+1)^{d^2-1}$.
\item The state $\omega_{A^n}$ is permutation invariant and invariant under $n$-fold product unitaries and commutes with all permutation invariant states.
\end{itemize}
\end{lemma}

The next lemma is a consequence of Lemma~\cite[Lemma 4.11]{tomamichel16_book} for $\alpha=1/2$. 
\begin{lemma}
\label{lem: pinching lemma}
Let \(\sigma_{X^mE^m}\) be a classical-quantum state, let \(m,k\geq 1\), and let
\(0\leq \ell<m\), with $n=mk+\ell$.
Let \(\omega_{E^m}\) be the universal state. Then,
\begin{align}
&-\log
F\!\left(
    \mathcal{R}_{h_n}(\sigma_{X^m E^m}^{\otimes k}
    \otimes
        \rho_{XE}^{\otimes\ell}),
    \frac{I_{Z_n}}{|\mathcal{Z}_n|}
    \otimes\rho_E^{\otimes n}
\right)\\
& \quad \qquad\leq
-\log
F\!\left(
    \mathcal{R}_{h_n}\!\left(
       \widetilde{\sigma}_{X^mE^m}^{\otimes k}
        \otimes
        \rho_{XE}^{\otimes\ell}
    \right),
    \frac{I_{Z_n}}{|\mathcal{Z}_n|}
    \otimes\rho_E^{\otimes n}
\right)
+k\log |\mathrm{spec}(\omega_{E^m})| .
\label{eq:omega-input-pinching}
\end{align}
where $\widetilde{\sigma}_{X^mE^m}= \mathcal{P}_{\omega_{E^m}}(\sigma_{X^mE^m})$. The same inequality remains valid when the pinching is performed with respect to $\rho_E^{\otimes m}$, with $|\mathrm{spec}(\omega_{E^m})|$ replaced by $|\mathrm{spec}(\rho_E^{\otimes m})|$.
\end{lemma}

\begin{proof}
Let us denote by $\Pi_{E^m,i}$ the spectral projectors of \(\omega_{E^m}\) and with $\lambda_i$ the corresponding eigenvalues. Explicitly, $\omega_{E^m}
    =
    \sum_{i=1}^{v_m}
    \lambda_i\Pi_{E^m,i}$. The map $
        \mathcal{P}_{\omega_{E^m}}
    ^{\otimes k}
    \otimes
    I_{E^\ell}$ defines a pinching map on the full system \(E^n\).
Indeed, its pinching projectors are
\begin{equation}
    \left\{
    \Pi_{E^m,i_1}
    \otimes\cdots\otimes
    \Pi_{E^m,i_k}
    \otimes I_{E^\ell}
    \right\}_{i_1,\ldots,i_k}\,.
\end{equation}
It has \(|\mathrm{spec}(\omega_{E^m})|^k\) pinching projectors. Applying the pinching inequality to the state
\(\sigma_{X^m E^m}\) gives
\begin{equation}
   \sigma_{X^m E^m}
    \leq
    |\mathrm{spec}(\omega_{E^m})|
    \left(
        I_{X^m}
        \otimes
        \mathcal{P}_{\omega_{E^m}}
    \right)
    (\sigma_{X^mE^m}).
\end{equation}
Taking the \(k\)-fold tensor power and tensoring with the unchanged
remainder \(\rho_{XE}^{\otimes\ell}\) yields
\begin{align}
    \sigma_{X^m E^m}^{\otimes k}\otimes \rho_{XE}^{\otimes \ell}
    &\leq
    |\mathrm{spec}(\omega_{E^m})|^k
    \left[
        \left(
            I_{X^m}
            \otimes
            \mathcal{P}_{\omega_{E^m}}
        \right)
        (\sigma_{X^mE^m})
    \right]^{\otimes k}
    \otimes\rho_{XE}^{\otimes\ell}.
\label{eq:partial-pinching-inequality}
\end{align}
Here, we used that if \(A\leq B\) are positive semidefinite, then
\(A^{\otimes k}\leq B^{\otimes k}\), which follows by induction.

Applying \(\mathcal{R}_{h_n}\) to
\eqref{eq:partial-pinching-inequality} and using that the hashing map acts only on the first register gives
\begin{align}
    \mathcal{R}_{h_n}(\sigma_{X^m E^m}^{\otimes k}\otimes \rho_{XE}^{\otimes \ell})
    &\leq
    |\mathrm{spec}(\omega_{E^m})|^k (I_{Z_n}
            \otimes
            \mathcal{P}_{\omega_{E^m}}^{\otimes k}\otimes I_{E^\ell})
    (\mathcal{R}_{h_n}\!(
        \sigma_{X^m E^m}^{\otimes k}\otimes \rho_{XE}^{\otimes \ell}
   ))
\label{eq:hashed-pinching-inequality}
\end{align}
Because
\(
[\omega_{E^m},\rho_E^{\otimes m}]=0
\),
we have
\begin{equation}
   (I_{Z_n}
            \otimes
            \mathcal{P}_{\omega_{E^m}}^{\otimes k}\otimes I_{E^\ell})
    \left(\frac{I_{Z_n}}{|\mathcal{Z}_n|}
    \otimes\rho_E^{\otimes n}\right)
    =  \frac{I_{Z_n}}{|\mathcal{Z}_n|}
    \otimes\rho_E^{\otimes n}.
\label{eq:ideal-state-invariant}
\end{equation}

We now use \eqref{eq:hashed-pinching-inequality}. We obtain
\begin{align}
&-\log
F\!\left(
    \mathcal{R}_{h_n}(\sigma_{X^m E^m}^{\otimes k}\otimes \rho_{XE}^{\otimes \ell}),
    \frac{I_{Z_n}}{|\mathcal{Z}_n|}
    \otimes\rho_E^{\otimes n}
\right)
\nonumber\\
&=
-2\log\Tr\!\Bigg[
    \left(
        \left(\frac{I_{Z_n}}{|\mathcal{Z}_n|}
    \otimes\rho_E^{\otimes n}\right)^{1/2}
        \mathcal{R}_{h_n}(\sigma_{X^m E^m}^{\otimes k}\otimes \rho_{XE}^{\otimes \ell})
        \left(\frac{I_{Z_n}}{|\mathcal{Z}_n|}
    \otimes\rho_E^{\otimes n}\right)^{1/2}
    \right)^{-1/2}\\
    &\qquad\qquad \qquad \qquad \times  
    \left(\frac{I_{Z_n}}{|\mathcal{Z}_n|}
    \otimes\rho_E^{\otimes n}\right)^{1/2}
    \mathcal{R}_{h_n}(\sigma_{X^m E^m}^{\otimes k}\otimes \rho_{XE}^{\otimes \ell})
    \left(\frac{I_{Z_n}}{|\mathcal{Z}_n|}
    \otimes\rho_E^{\otimes n}\right)^{1/2}
\Bigg]\\
&\leq 
-2\log\Tr\!\Bigg[
    \left(
        \left(\frac{I_{Z_n}}{|\mathcal{Z}_n|}
    \otimes\rho_E^{\otimes n}\right)^{1/2}
\mathcal{P}_{\omega_{E^m}}^{\otimes k}
    (\mathcal{R}_{h_n}\!(
        \sigma_{X^m E^m}^{\otimes k}\otimes \rho_{XE}^{\otimes \ell}
   ))
        \left(\frac{I_{Z_n}}{|\mathcal{Z}_n|}
    \otimes\rho_E^{\otimes n}\right)^{1/2}
    \right)^{-1/2} \\
    &\qquad\qquad  \times  
    \left(\frac{I_{Z_n}}{|\mathcal{Z}_n|}
    \otimes\rho_E^{\otimes n}\right)^{1/2}
    \mathcal{R}_{h_n}(\sigma_{X^m E^m}^{\otimes k}\otimes \rho_{XE}^{\otimes \ell})
    \left(\frac{I_{Z_n}}{|\mathcal{Z}_n|}
    \otimes\rho_E^{\otimes n}\right)^{1/2}
\Bigg] +k\log |\mathrm{spec}(\omega_{E^m})|\\
&= 
-2\log\Tr\!\Bigg[
    \left(
        \left(\frac{I_{Z_n}}{|\mathcal{Z}_n|}
    \otimes\rho_E^{\otimes n}\right)^{1/2}
\mathcal{P}_{\omega_{E^m}}^{\otimes k}
    (\mathcal{R}_{h_n}\!(
        \sigma_{X^m E^m}^{\otimes k}\otimes \rho_{XE}^{\otimes \ell}
   ))
        \left(\frac{I_{Z_n}}{|\mathcal{Z}_n|}
    \otimes\rho_E^{\otimes n}\right)^{1/2}
    \right)^{-1/2} \\
    &\qquad \times  
    \left(\frac{I_{Z_n}}{|\mathcal{Z}_n|}
    \otimes\rho_E^{\otimes n}\right)^{1/2}
    \mathcal{P}_{\omega_{E^m}}^{\otimes k}
    (\mathcal{R}_{h_n}\!(
        \sigma_{X^m E^m}^{\otimes k}\otimes \rho_{XE}^{\otimes \ell}
   ))
    \left(\frac{I_{Z_n}}{|\mathcal{Z}_n|}
    \otimes\rho_E^{\otimes n}\right)^{1/2}
\Bigg] +k\log |\mathrm{spec}(\omega_{E^m})|\\
&=
-\log
F\!\left(
    \mathcal{P}_{\omega_{E^m}}^{\otimes k}
    (\mathcal{R}_{h_n}\!(
        \sigma_{X^m E^m}^{\otimes k}\otimes \rho_{XE}^{\otimes \ell}
   )),
    \frac{I_{Z_n}}{|\mathcal{Z}_n|}
    \otimes\rho_E^{\otimes n}
\right)
+ k\log |\mathrm{spec}(\omega_{E^m})| .
\end{align}
In the inequality, we used the operator antimonotonicity of
\(x\mapsto x^{-1/2}\) together with
\eqref{eq:hashed-pinching-inequality}.

It remains to justify the second equality. The pinching
\( I_{Z_n}
            \otimes
            \mathcal{P}_{\omega_{E^m}}^{\otimes k}\otimes I_{E^\ell}\) has projectors
\begin{equation}
    \left\{
    I_{Z_n}
    \otimes
    \Pi_{E^m,i_1}
    \otimes\cdots\otimes
    \Pi_{E^m,i_k}
    \otimes I_{E^\ell}
    \right\}_{i_1,\ldots,i_k}.
\end{equation}
By \eqref{eq:ideal-state-invariant},
\(\frac{I_{Z_n}}{|\mathcal{Z}_n|}
    \otimes\rho_E^{\otimes n}\) is block diagonal with respect to these projectors.
Moreover,
\begin{align}
    (I_{Z_n}
            \otimes
            \mathcal{P}_{\omega_{E^m}}^{\otimes k}\otimes I_{E^\ell})
    \mathcal{R}_{h_n}(
        \sigma_{X^m E^m}^{\otimes k}\otimes \rho_{XE}^{\otimes \ell})
\end{align}
is block diagonal by definition. Consequently,
\begin{equation}
   \left(
        \left(\frac{I_{Z_n}}{|\mathcal{Z}_n|}
    \otimes\rho_E^{\otimes n}\right)^{1/2}
\mathcal{P}_{\omega_{E^m}}^{\otimes k}
    (\mathcal{R}_{h_n}\!(
        \sigma_{X^m E^m}^{\otimes k}\otimes \rho_{XE}^{\otimes \ell}
   ))
        \left(\frac{I_{Z_n}}{|\mathcal{Z}_n|}
    \otimes\rho_E^{\otimes n}\right)^{1/2}
    \right)^{-1/2}
\end{equation}
is block diagonal with respect to the same projectors. Inserting the
corresponding resolution of the identity inside the trace and using
cyclicity therefore gives the desired identity. Finally, since hashing acts on the first register, we can commute the pinching with the hashing map to conclude the proof.
\end{proof}

\end{document}